\documentclass[11pt]{article}
\usepackage{amsmath,amssymb,amsthm}
\usepackage{booktabs}
\newcommand{\dif}{\mathrm{d}}
\usepackage[margin=1in]{geometry}
\usepackage{hyperref}
\usepackage{graphicx}
\usepackage{capt-of}
\usepackage{comment}

\newtheorem{theorem}{Theorem}[section]
\newtheorem{proposition}{Proposition}[section]
\newtheorem{corollary}{Corollary}[section]
\newtheorem{lemma}{Lemma}[section]
\newtheorem{example}{Example}[section]
\newtheorem{remark}{Remark}[section]
\newtheorem{definition}{Definition}[section]

\newcommand{\E}{\mathbb{E}}
\newcommand{\be}{\mbox{\boldmath $e$}}
\newcommand{\Var}{\mathrm{Var}}
\newcommand{\Cov}{\mathrm{Cov}}
\newcommand{\spr}{\mathrm{spr}}
\newcommand{\dfn}{\stackrel{\Delta} {=}}

\title{Statistics of the Compression Ratio of a Variable-to-Variable Code:\\
Exact Moments and Asymptotic Behavior\\
\large(Extended version, with full proofs)}
\author{Neri Merhav}

\begin{document}
\maketitle
\thispagestyle{empty}

\begin{center}
The Viterbi Faculty of Electrical and Computer Engineering\\
Technion - Israel Institute of Technology\\
Technion City, Haifa 3200003, ISRAEL\\
E--mail: {\tt merhav@technion.ac.il}\\
\end{center}
\vspace{1.5\baselineskip}
\setlength{\baselineskip}{1.5\baselineskip}

\begin{abstract}
A variable-to-variable (V2V) length code parses a source sequence into
phrases of variable length and maps each phrase to a binary codeword of,
generally, a different random length. After encoding $n$ phrases, the
realized compression ratio $R_n=\Lambda_n/\Sigma_n$ -- total codeword
length over total source-symbol count -- is the finite-sample counterpart
of the code's asymptotic rate $\rho$, to which it converges only as
$n\to\infty$. This paper first derives exact formulas for all integer moments
of $R_n$ for a given discrete memoryless source (DMS). 
Specifically, we obtain a closed-form formula for every moment $\E\{R_n^k\}$ as a
one-dimensional integral involving only single-phrase moment
generating functions of
the pair $(L,\ell)$ -- the phrase length, in source symbols, and
codeword length, in bits. From these moments we derive an Edgeworth
approximation to the cumulative distribution function (CDF) of $R_n$
that is substantially more accurate than the central
limit theorem (CLT) approximation. Using the Laplace method of integration, 
we also derive explicit closed-form formulas 
for the bias constant $C=\lim_{n\to\infty}n(\E\{R_n\}-\rho)$ and for 
the variance constant $\lim_{n\to\infty}n\cdot\Var\{R_n\}$. The
analysis extends to Markov sources via state-indexed matrices with
a redundancy formula obtained in closed form.

On the coding-theoretic side, we cast V2V length codes as finite-state
encoders and apply a generalized Kraft inequality for a
compression-rate lower bound, and give a structural decomposition of
the bias coefficient that separates cleanly across variable-to-fixed (V2F) length
codes, fixed-to-variable (F2V) length codes, and
V2V length codes. Applied to the Khodak code of Bugeaud, Drmota,
and Szpankowski, this decomposition shows that its improved
performance is reflected in its smaller bias
constant.
\end{abstract}

\section{Introduction}

\subsection{Objectives and motivation}

A variable-to-variable (V2V) length source code is the most general
member of the family of lossless codes considered here. It parses the
source into phrases of variable length and encodes each phrase with a
binary codeword of variable length, combining the structural advantages of
variable-to-fixed
(V2F) length coding -- which adapts the parsing to source statistics -- with
those of fixed-to-variable (F2V) length coding -- which adapts codeword lengths
to symbol probabilities.

The standard performance measure for a V2V length code is its asymptotic
compression rate, $\rho = \E\{\ell\}/\E\{L\}$, the ratio of expected
codeword length to expected phrase length. Let $\Lambda_n$ and
$\Sigma_n$ denote, respectively, the total codeword length (in bits)
and the total source-symbol count after $n$ encoded phrases. By the
strong law of large
numbers ($\Lambda_n/n\to\E\{\ell\}$ and $\Sigma_n/n\to\E\{L\}$ almost surely) and the
continuous mapping theorem \cite{vandervaart}, the realized compression ratio $R_n =
\Lambda_n/\Sigma_n$ converges to $\rho$ almost surely as $n\to\infty$. This
asymptotic rate, however, answers neither how quickly convergence occurs
nor how large the fluctuations around $\rho$ are at any finite number of
encoded phrases $n$. Both questions are practically important. In buffer
analysis for variable-length codes, the fluctuation of the compression
ratio over a finite window directly determines buffer overflow
probabilities. In code design, comparing candidate codes requires
evaluating their performance at finite $n$, not merely their limit.

We address these questions by deriving exact formulas for all integer
moments of $R_n$, for a given V2V length code and a memoryless source. This is
fundamentally an analysis of a \emph{given} code's finite-sample
behavior, not a code-design method; the joint design of the parsing
tree and codeword lengths for a target $n$ remains open (see
Section~\ref{sec:design}). Our main
result is an exact formula for $\E\{R_n^k\}$ for every positive integer $k$ as a
single one-dimensional integral, requiring only the single-phrase
joint moment generating function (MGF) of $(L,\ell)$.
The underlying tool is an integral
representation of $1/\Sigma_n$ as well as its $k$-th power as
a Laplace transform, which enables turning a ratio of sums into a one-dimensional integral of
simple single-phrase quantities.
The novelty is in the systematic extension to
\emph{every} integer moment via the set-partition combinatorics of
Theorem~\ref{thm:general}; the fact that the resulting formulas are
exact at every finite $n$ rather than asymptotic; and the further
extension to Markov sources (Section~\ref{sec:markov}), which replaces
scalar moment generating functions by matrix-valued ones and requires a corresponding
eigenvalue-perturbation argument with no scalar analogue.
From these moments we derive closed-form
asymptotic formulas for the bias constant
$\lim_{n\to\infty}n(\E\{R_n\}-\rho)$
and the variance constant
$\lim_{n\to\infty} n\Var\{R_n\}$, and we construct an 
Edgeworth approximation to the CDF of $R_n$ that is substantially more accurate than the
central limit theorem (CLT). Before turning to this analysis, Section~\ref{sec:background} records
some necessary background, including a compression-rate lower bound
obtained by observing that a V2V length code is an instance of the finite-state
encoders covered by the generalized Kraft inequality of \cite{merhav-gki};
this observation sets the stage for the moment analysis but is logically
independent of it. Unlike the classical route to such a bound --
concatenating many phrases and invoking a law-of-large-numbers argument
to identify the limiting rate -- the generalized-Kraft-inequality bound
is a direct, purely algebraic consequence of a single spectral-radius
inequality, with no block-length limit theorem involved.

\subsection{Related work}

The V2F and V2V length coding literature is relatively sparse. Existing work,
surveyed in \cite{drmota-szpan}, focuses on the average redundancy of
specific named codes (Huffman, Tunstall, Khodak, Boncelet) for known
sources, using analytic combinatorics and Mellin-transform 
methods; Savari and Szpankowski \cite{savari-szpan} give an early
analysis specifically of V2V length codes along these lines. The moments of $R_n$ at finite $n$, as distinct from
the asymptotic rate $\rho$, have not been addressed.

\paragraph{V2F length coding.}
Tunstall's algorithm \cite{tunstall} constructs an optimal uniquely
parsable V2F dictionary by a simple greedy procedure. Savari and Gallager
\cite{savari-gallager} established the asymptotic redundancy of Tunstall
codes using renewal theory (modeling self-information as a regenerative
process) and Markov reward machinery; their analysis brings in second
moments of inter-renewal times, but only as a correction to the precision
of an asymptotic mean formula, not as a variance target. Drmota, Reznik,
and Szpankowski \cite{drs} established a CLT and a variance formula for the
V2F phrase length $L$, via renewal theory and Mellin transforms, as the
dictionary size $M\to\infty$. Both bodies of work address the phrase
length $L$ alone, under a fixed-length codeword assignment, as a function
of a growing dictionary; our asymptotic variable is the number of phrases
$n$ processed by a fixed code, and our object is the ratio $R_n =
\Lambda_n/\Sigma_n$ rather than $L$.

\paragraph{Plurally parsable dictionaries.}
Savari \cite{savari-pp, savari-survey} showed that relaxing unique
parsability can outperform the Tunstall dictionary of the same size for
highly predictable sources, with exact results for specific binary-source
families but no general algorithm.

\paragraph{Delay and redundancy.}
Shayevitz, Meron, Feder, and Zamir \cite{smfz} characterize the
redundancy--delay trade-off for V2V and related code families. Their
object is the expected per-symbol code length as a function of a
worst-case delay constraint, not the moments of $R_n$ as the number of phrases grows.

\paragraph{Codeword-length distributions.}
Courtade and Verd\'u \cite{cv} study the CGF and Gaussian
approximation of the codeword length of an optimal F2V length code on an
$n$-symbol block, tracing back to Strassen's foundational second-order
asymptotic analysis \cite{strassen62} and complementing
Kontoyiannis's earlier second-order
noiseless source coding theorems \cite{kontoyiannis97} (which also
covers Markov sources) and the non-asymptotic refinements of
Kontoyiannis and Verd\'u \cite{kontoyiannis-verdu14}. This is a single i.i.d.\ sum, not a ratio of two
correlated sums with a random denominator; the ratio structure is specific
to V2V and V2F coding and is the source of the combinatorial complexity
addressed in the present paper.

\paragraph{Generalized Kraft inequality.}
Merhav \cite{merhav-gki} establishes a generalized Kraft inequality for
finite-state encoders, strictly improving the earlier result of Ziv and
Lempel \cite{ziv-lempel}. We use this to derive an explicit lower bound on
the V2V compression ratio in terms of the parsing tree structure.

\section{System model and background on V2V codes}

\label{sec:background}
\subsection{Source model and notation}

We consider a discrete memoryless source (DMS) $X_1,X_2,\dots$, where
each symbol $X_i$ ($i$ -- positive integer) is a random variable taking values
in a finite alphabet $\mathcal{X}$ of cardinality $\alpha$; specific
realizations are denoted by $x_1,x_2,\dots$. Each symbol is drawn
according to the source distribution $P=\{P(x),~x\in\mathcal{X}\}$.
The single-letter entropy is
\begin{equation}
\label{eq:entropy}
H(X) = -\sum_{x\in\mathcal{X}} P(x)\log_2 P(x).
\end{equation}
A V2V length code is specified by two components:
\begin{enumerate}
\item[1)] A \emph{source dictionary} $\mathcal{D}$: a prefix-free set
of variable-length source strings, represented by the leaves of a
complete $\alpha$-ary tree. The size of the dictionary is
$|\mathcal{D}|=M$, and a generic member of $\mathcal{D}$ is denoted by
$y$. A source string $x_1,x_2,\dots$ is parsed by the dictionary
parser into a succession of phrases $y_1,y_2,\dots$; the corresponding
random variables $Y_1,Y_2,\dots$, induced by the source $P$, form
another DMS with alphabet $\mathcal{D}$ (of size $M$) and phrase
probabilities $\{Q(y),\,y\in\mathcal{D}\}$, where each $Q(y)$ is given by the product of
the letter probabilities that make up the phrase $y$. Let $L(y)$,
$y\in\mathcal{D}$, denote the length of phrase $y$, in source symbols.
We denote by $H(Y)\dfn-\sum_{y\in\mathcal{D}}Q(y)\log_2 Q(y)$ the phrase entropy.

\item[2)] A variable-length uniquely decodable (UD) code, mapping
$\mathcal{D}$ into a set of variable-length binary strings. Let
$\ell(y)$ denote the length, in bits, of the codeword assigned to
$y\in\mathcal{D}$.
\end{enumerate}

We index the successive phrases produced by the parser by
$i=1,2,\dots$; the $i$-th phrase is $Y_i$, of length $L(Y_i)$
source symbols, encoded by a codeword of length $\ell(Y_i)$
bits. For a DMS, the pairs $(L(Y_i),\ell(Y_i))_{i\ge1}$ are
independent and identically distributed; we write $Y$ for a generic
random phrase with distribution $Q$, so that $L(Y)$ and
$\ell(Y)$ denote its length and its codeword length, respectively. When it causes no
ambiguity, we abbreviate $L(Y)$ and $\ell(Y)$ by $L$ and $\ell$; the
explicit argument $Y$ is retained whenever a formula defines a new
named quantity as an expectation. The
following single-phrase quantity appears throughout:
\begin{equation}
\mu_j(t)\dfn \E\{[\ell(Y)]^j e^{-tL(Y)}\},~j=0,1,2,\dots,~t\ge0. \label{eq:muj}
\end{equation}
Note that $\mu_0(t)=\E\{e^{-tL(Y)}\}$ is the MGF of
$-L(Y)$. We define the
mean phrase length and mean codeword length by
\begin{equation}
\label{eq:barLell}
\bar{L}\dfn\E\{L(Y)\}=\sum_{y\in\mathcal{D}} Q(y)L(y),
\qquad
\bar{\ell}\dfn\E\{\ell(Y)\}=\sum_{y\in\mathcal{D}} Q(y)\ell(y).
\end{equation}
The asymptotic compression rate of the code is defined as
\begin{equation}
\rho \dfn \frac{\bar\ell}{\bar L} = -\frac{\mu_1(0)}{\mu_0'(0)}, 
\end{equation}
where $\mu_0'(\cdot)$ denotes the derivative of $\mu_0(\cdot)$.
For any two random
variables $U,V$ appearing below, we write $\Var\{U\}\dfn\E\{U^2\}-(\E\{U\})^2$
and $\Cov\{U,V\}\dfn\E\{UV\}-\E\{U\}\E\{V\}$ for their variance and
covariance, respectively.

After encoding $n$ phrases, the total code length is
$\Lambda_n \dfn \sum_{i=1}^n \ell(Y_i)$ (in bits) and the total number of
source symbols consumed is $\Sigma_n \dfn \sum_{i=1}^n L(Y_i)$. The
\emph{realized compression ratio} over $n$ phrases is
\begin{equation}
\label{eq:Rndef}
R_n \dfn \frac{\Lambda_n}{\Sigma_n} = \frac{\sum_{i=1}^n
\ell(Y_i)}{\sum_{i=1}^n L(Y_i)}~~\frac{\mbox{bits}}{\mbox{symbol}}.
\end{equation}
By the strong law of large numbers,
$\Lambda_n/n\to\bar\ell$ and $\Sigma_n/n\to\bar L$ almost surely (a.s.),
and so, by the continuous mapping theorem, $R_n\to\rho$ a.s.\ as
$n\to\infty$.

\subsection{Structure and parsing}

The dictionary $\mathcal{D}$ of a V2V length code is represented by a full
$\alpha$-ary rooted tree, i.e., every internal node has 
$\alpha$ children, one per each possible source symbol; the $M$ leaves correspond to
the $M$ phrases $y\in\mathcal{D}$. The parsing rule is as follows: start at
the root, follow the edge labeled by each successive source symbol
until reaching a leaf. The leaf reached identifies the current phrase
$y$; the encoder emits the binary codeword assigned to that leaf and
resets to the root for the next phrase.

Let $J$ denote the number of internal nodes of the parsing tree,
including the root. Unique parsability forces $M=J(\alpha-1)+1$, so
\begin{equation}
\label{eq:Jdef}
J = \frac{M-1}{\alpha-1}.
\end{equation}
The phrase length satisfies $L(y)\ge1$ for every $y\in\mathcal{D}$, and
the codeword length $\ell(y)$ takes values in
$\{1,2,\dots,\ell_{\max}\}$ for some maximum codeword length
$\ell_{\max}$.

\subsection{V2V as a finite-state encoder}

A \emph{finite-state (FS) encoder} over the source alphabet $\mathcal{X}$
and output alphabet $\mathcal{B}$ is specified by a finite state set
$\mathcal{Z}$ (of size $|\mathcal{Z}|$), a \emph{next-state function} $g:\mathcal{Z}\times\mathcal{X}
\to\mathcal{Z}$, and an \emph{output function} $f:\mathcal{Z}\times\mathcal{X}
\to\mathcal{B}^*$, where $\mathcal{B}^*$ denotes a set of finite
strings over $\mathcal{B}$, possibly including the empty string $\emptyset$ of
length zero. 
For a string $s\in\mathcal{B}^*$, we write $\ell[s]$ for its length in
bits, with the convention $\ell[\emptyset]=0$; this is consistent with
$\ell(y)$ of Section~\ref{sec:background} when $s=c(y)$ is the codeword
of leaf $y$. 
When a source sequence
$x_1,x_2,\dots$ is fed into the encoder, the state evolves according to
the recursion $z_{i+1} = g(z_i,x_i)$, $i=1,2,\dots$ (with $z_1$ being a fixed
initial state), and the encoder outputs the sequence of strings from
$\mathcal{B}$,
$f(z_1,x_1), f(z_2,x_2),\dots$

The V2V length encoder described in the preceding subsection can be cast as
an FS encoder with a state set $\mathcal{Z}$ given by the set
of $J$ internal nodes of the parsing tree (thus $|\mathcal{Z}|=J$), source alphabet
$\mathcal{X}$ (of cardinality $\alpha$), and output alphabet
$\mathcal{B}=\{0,1\}$ (binary codewords). The encoder is in state
$z\in\mathcal{Z}$ when it has partially matched a phrase and is currently
at internal node $z$. For source symbol $x$ processed while in state $z$,
the next-state and output functions are:
\begin{equation}
\label{eq:nextstate}
g(z,x) = \begin{cases}
z' & \text{if the child of $z$ along edge $x$ is an internal node $z'$,}\\
\mbox{root} & \text{if the child of $z$ along edge $x$ is a leaf,}
\end{cases}
\end{equation}
\begin{equation}
\label{eq:output}
f(z,x) = \begin{cases}
\emptyset & \text{ if }g(z,x)\ne \mbox{root},\\
c(y) & \text{if the child of $z$ along $x$ is leaf $y$.}
\end{cases}
\end{equation}
The encoder emits output only when a phrase is completed (i.e., when the
child is a leaf), at which point it resets to the root.

\begin{definition}[Information lossless encoder, \cite{ziv-lempel}]
\label{def:IL}
An FS encoder is \emph{information lossless (IL)} if, given the initial
state, the output string, and the final state, the input string can be
uniquely recovered.
\end{definition}

The IL property holds for uniquely parsable dictionaries, since each
codeword maps to a unique leaf, which together with the initial parsing
state determines the phrase consumed. For plurally parsable dictionaries
\cite{savari-pp} -- where a source string may have more than one valid
parse -- the IL property may fail, and the framework below requires
modification.

For $m\ge1$, $g(z,x^m)$ will denote the final state of the encoder when
the initial state is $z$ and it processes the $m$ successive inputs
$x^m=(x_1,\dots,x_m)\in\mathcal{X}^m$; $f(z,x^m)$ denotes the
corresponding output string emitted in response to $x^m$.

\subsection{The generalized Kraft inequality and a lower bound on the compression ratio}

The classical Kraft--McMillan inequality, $\sum_{y} 2^{-\ell(y)}\le1$,
characterizes a limitation on lossless codes over a fixed alphabet but does not directly
apply to FS encoders. For the IL FS encoder
above, the appropriate generalization uses the following matrix.

\begin{definition}[Kraft matrix, \cite{merhav-gki}]
For an IL FS encoder with state space $\mathcal{Z}$, the Kraft matrix
$\mathbf{K}\in\mathbb{R}^{|\mathcal{Z}|\times|\mathcal{Z}|}$ has entries
\begin{equation}
\label{eq:Kdef}
\mathbf{K}_{zz'} = \sum_{\{x:~g(z,x)=z'\}} 2^{-\ell[f(z,x)]}, ~~~z,z'\in\mathcal{Z},
\end{equation}
\end{definition}

Let $\spr(\mathbf{K})$ denote the spectral radius of $\mathbf{K}$. The
generalized Kraft inequality reads:

\begin{theorem}[Generalized Kraft inequality, \cite{merhav-gki}]
\label{thm:gki}
For every IL FS encoder, $\spr(\mathbf{K})\le1$.
\end{theorem}

\begin{theorem}[Irreducible bound, \cite{merhav-gki}]
\label{thm:irred}
If $\mathbf{K}$ is irreducible, then for all $z,z'\in\mathcal{Z}$ and every
integer $m\ge1$,
\begin{equation}
\label{eq:Kbound}
(\mathbf{K}^m)_{zz'}=
\sum_{\substack{x^m:~g(z,x^m)=z'}} 2^{-\ell[f(z,x^m)]}
\le 2^{(J-1)\ell_{\max}}.
\end{equation}
\end{theorem}
Note that the right-hand side is \emph{independent of $m$}. This is the key
improvement over the earlier generalized Kraft inequality of Ziv and
Lempel \cite{ziv-lempel}, whose analogous bound grows linearly in $m$;
this linear growth prevents one from extracting any useful rate lower
bound from the Ziv--Lempel result.

Substituting $J=(M-1)/(\alpha-1)$ into Theorem~\ref{thm:irred} and
taking the supremum over $m\ge1$ gives the following lower bound on
the compression ratio of any V2V length code, valid for any source (not
necessarily memoryless); this is eq.~(34) of \cite{merhav-gki}, whose
proof combines the per-entry bound of Theorem~\ref{thm:irred} with the
standard Kraft-sum-with-slack redundancy bound and is not reproduced
here:

\begin{corollary}[Lower bound on compression ratio]
\label{cor:lower}
For an irreducible, uniquely parsable V2V length code with $M$ leaves over an
$\alpha$-ary source alphabet and maximum codeword length $\ell_{\max}$,
the asymptotic compression rate $\rho$ satisfies
\begin{equation}
\label{eq:lowerbound}
\rho\ge\sup_{m\ge1}\Bigl[H(X_m\mid X^{m-1}) - \frac{C(M,\alpha,\ell_{\max})}{m}\Bigr],
\end{equation}
where $H(X_m\mid X^{m-1})$ is the conditional entropy (in bits per
source symbol) of the $m$-th source symbol given the preceding block
$X^{m-1}:=(X_1,\dots,X_{m-1})$, and
\begin{equation}
\label{eq:Cdef}
C(M,\alpha,\ell_{\max})\dfn 
2\log_2 J + (J-1)\ell_{\max}=
2\log_2\frac{M-1}{\alpha-1} +
\frac{(M-\alpha)\ell_{\max}}{\alpha-1}.
\end{equation}
\end{corollary}
The right-hand side of \eqref{eq:lowerbound} cannot be smaller than the entropy
rate $H_\infty\dfn\lim_{m\to\infty}H(X_m\mid X^{m-1})$ (a supremum
dominates the limit), and decreases in $M$.
For a DMS, $H(X_m\mid X^{m-1})=H(X)$,
so the supremum is taken over
$H(X)-C(M,\alpha,\ell_{\max})/m$, which yields $H(X)$.
This is a substantially
more direct route to the converse than the classical technique of
concatenating many phrases and invoking a law-of-large-numbers
argument. 

\subsection{Design}
\label{sec:design}

Unlike F2V length coding (where Huffman's algorithm gives a provably optimal
greedy construction) and V2F length coding (where Tunstall's algorithm does the
same for memoryless sources), no analogous optimal joint design algorithm is known
for V2V length coding, where both the parsing tree and the codeword lengths are
simultaneously free. The standard practice is a two-stage heuristic:
build a Tunstall-optimal V2F tree, then apply Huffman coding to the leaf
probabilities induced by that tree. This heuristic is not known to be
jointly optimal and in general it is not \cite{bugeaud08}. The minimum
achievable redundancy for V2V length codes of a given size $M$ is an open problem.
One known negative result: for a binary memoryless source with
$0<P(1)<1/2$, no V2V code achieves exactly zero redundancy \cite{drmota-szpan}.

The best known achievability result is due to Khodak \cite{drmota-szpan,bugeaud08}:
for any memoryless source there exists a V2V code whose average redundancy
decays as $O(\bar L^{-5/3})$, where $\bar L$ is the average phrase
length.  This strictly improves on the $O(1/\bar L)$ redundancy of
the Tunstall--Huffman heuristic, achieved by using a parsing tree that
deliberately maintains non-uniform leaf probabilities so that variable
codeword lengths can contribute meaningfully.  An explicit construction
achieving this bound is given in \cite{bugeaud08}.

\subsection{Notation summary}
\label{sec:notation}

The paper accumulates a fair amount of notation as it goes; the following table
collects the symbols used across more than one section, for reference
while reading the more technical parts (Sections~\ref{sec:moments}
and~\ref{sec:markov} especially). Purely local symbols, introduced and
used within a single proof, are omitted.

\begingroup
\small
\begin{center}
\begin{tabular}{@{}lll@{}}
\toprule
Symbol & Meaning & First defined \\
\midrule
\multicolumn{3}{@{}l}{\emph{Source and code structure}} \\
$\mathcal X$, $\alpha$ & source alphabet, its cardinality & \S\ref{sec:background} \\
$\mathcal D$, $M$ & dictionary (leaf set), its size & \S\ref{sec:background} \\
$Y_i$, $L(Y_i)$, $\ell(Y_i)$ & $i$-th phrase, its length, its codeword length & \S\ref{sec:background} \\
$H(X)$ & source symbol entropy & \S\ref{sec:background}\\
$H(Y)$ & source phrase entropy & \S\ref{sec:background}\\
$\bar L,\bar\ell$ & mean phrase length, mean codeword length & \S\ref{sec:background} \\
$\rho$ & asymptotic compression rate, $\bar\ell/\bar L$ & \S\ref{sec:background} \\
$\Lambda_n,\Sigma_n,R_n$ & total codeword length, total symbols, $R_n=\Lambda_n/\Sigma_n$ & \S\ref{sec:background} \\
$\mathbf K$, $\spr(\mathbf K)$ & Kraft matrix, its spectral radius & \S\ref{sec:background} \\
\multicolumn{3}{@{}l}{\emph{Memoryless moment analysis (Section~\ref{sec:moments})}} \\
$\mu_j(t)$ & single-phrase quantity $\E\{[\ell(Y)]^je^{-tL(Y)}\}$ & \S\ref{sec:background} \\
$f(t)$ & $-\ln\mu_0(t)$, negative log of the MGF of $-L(Y)$ & \S\ref{sec:laplace} \\
$\kappa_2$, $c_{\ell L}$ & $\Var\{L(Y)\}$, $\E\{\ell(Y)L(Y)\}$ & \S\ref{sec:laplace} \\
$C$ & $O(1/n)$ bias coefficient, $\E\{R_n\}\to\rho+C/n$ & \S\ref{sec:laplace} \\
$\gamma_1$ & skewness of $R_n$ & \S\ref{sec:all-moments} \\
\multicolumn{3}{@{}l}{\emph{Markov extension (Section~\ref{sec:markov})}} \\
$X_k$, $S_i$ & underlying symbol chain, phrase-boundary state $X_{\Sigma_i}$ & \S\ref{sec:markov_model} \\
$\boldsymbol\mu_j(t)$ & matrix analogue of $\mu_j(t)$, indexed by $\mathcal X$ & \S\ref{sec:markov_transforms} \\
$\pi$ & stationary distribution of $\boldsymbol\mu_0(0)$ (row vector) & \S\ref{sec:markov_transforms} \\
$\E_\pi\{Q\}$ & $\sum_s\pi_s\E\{Q(Y_i)\mid S_{i-1}=s\}$ & \S\ref{sec:markov_transforms} \\
$\Lambda(t)$ & dominant (Perron) eigenvalue of $\boldsymbol\mu_0(t)$ & \S\ref{sec:markov_worked} \\
$\bar c$ & $\E_\pi\{L\ell\}$ (and $\bar L,\bar\ell$ reused for
$\E_\pi\{L\},\E_\pi\{\ell\}$) & \S\ref{sec:markov_bias} \\
$g_s$ & $\E\{L(Y_i)\mid S_{i-1}=s\}$ & \S\ref{sec:markov_worked} \\
$g_\ell(s)$ & $\E\{\ell(Y_i)\mid S_{i-1}=s\}$ & \S\ref{sec:markov_bias} \\
$w$, $w_\ell$ & Poisson-equation solutions for the rewards $L$, $\ell$ & \S\ref{sec:markov_bias} \\
\bottomrule
\end{tabular}
\end{center}
\endgroup

\section{Moments of the compression ratio for memoryless sources}
\label{sec:moments}

This section develops the main technical machinery of this paper in two
stages. In Sections~\ref{sec:mean-moment}--\ref{sec:all-moments}, an exact
formula for every integer moment $\E\{R_n^k\}$ is obtained.
First, the case $k=1$ is worked out in full, and 
then it is generalized for an arbitrary positive integer $k$. Second, from these exact formulas: a
refined Edgeworth approximation to the CDF of $R_n$
(Section~\ref{sec:edgeworth}) that improves on the classical CLT by using
the skewness the third moment provides; closed-form $O(1/n)$
asymptotics for the bias and variance via Laplace's method
(Section~\ref{sec:laplace}); and a large-deviations analysis of $R_n$'s
tails. A numerical validation against Monte Carlo simulation
(Section~\ref{sec:validation}) checks the exact formulas along the way.
In this section the phrase pairs $(L(Y_i),\ell(Y_i))_{i\ge1}$ are i.i.d.\ with the
same distribution as $(L(Y),\ell(Y))$ for a generic phrase $Y$, as is appropriate for a
DMS. 

\subsection{The mean compression ratio}
\label{sec:mean-moment}

We begin with the first moment -- the expected compression ratio $\E\{R_n\}$ -- as
the quantity of primary practical interest and the one that motivates the
general approach.

\begin{theorem}
\label{thm:mean}
For a given DMS and a given V2V length code with a dictionary $\mathcal{D}$
and code length function $\ell(\cdot)$, the expected compression ratio of $n$
phrases is given by
\begin{equation}
\E\{R_n\}= n\cdot\int_0^\infty \mu_1(t)\cdot[\mu_0(t)]^{n-1}\dif t.
\label{eq:mean}
\end{equation}
\end{theorem}

\begin{proof}
We use the integral representation $1/s$ as the Laplace 
transform of the unit step function, that is, the identity
$\frac{1}{s} = \int_0^\infty e^{-st}\dif t$,
valid for any $s>0$. Since $\Sigma_n = \sum_{i=1}^n L(Y_i)$ is a positive integer
almost surely, we have
\begin{equation}
\label{eq:Rn_integrand}
R_n = \frac{\Lambda_n}{\Sigma_n} = \Lambda_n\cdot \int_0^\infty e^{-t\Sigma_n}\dif t
\quad \text{a.s.}
\end{equation}
Taking expectations (justified here by finite-sum linearity, since
$\Lambda_n$ and $\Sigma_n$ take only finitely many values for the
finite dictionaries considered throughout this paper, so the outer
expectation is itself a finite sum that commutes with the integral
term by term), we obtain
\begin{equation}
\label{eq:ERn_first}
\E\{R_n\} = \int_0^\infty \E\{\Lambda_n e^{-t\Sigma_n}\}\dif t.
\end{equation}
Now expand 
\begin{equation}
\Lambda_n e^{-t\Sigma_n} = \bigl[\sum_{i=1}^n \ell(Y_i)\bigr]\cdot
\exp\bigl[-t\sum_{i=1}^n L(Y_i)\bigr]
= \sum_{i=1}^n \ell(Y_i) \prod_{j=1}^n e^{-tL(Y_j)}. 
\end{equation}
Taking the expectation of
the $i$-th term and using independence of the phrase pairs:
\begin{equation}
\label{eq:independence}
\E\Bigl\{\ell(Y_i)\prod_{j=1}^n e^{-tL(Y_j)}\Bigr\}
= \E\{\ell(Y_i) e^{-tL(Y_i)}\}\prod_{j\ne i}\E\{e^{-tL(Y_j)}\}
= \mu_1(t)[\mu_0(t)]^{n-1}.
\end{equation}
Since all $n$ terms contribute equally (by the i.i.d.\ assumption),
$\E\{\Lambda_n e^{-t\Sigma_n}\}=n\cdot\mu_1(t)[\mu_0(t)]^{n-1}$. Substituting
this into \eqref{eq:ERn_first} and integrating over $t$ from $0$ to
$\infty$,
\begin{equation}
\E\{R_n\}=n\cdot\int_0^\infty \mu_1(t)\cdot[\mu_0(t)]^{n-1}\dif t,
\end{equation}
which is \eqref{eq:mean}.
\end{proof}

\subsection{Extension to general moments}
\label{sec:all-moments}

The proof of Theorem~\ref{thm:mean} used two ingredients: the Laplace
transform representation of $1/s$, and independence across phrase pairs. For
general $k\ge1$, the same two ingredients apply, but expanding
$\Lambda_n^k$ now produces cross-terms between different phrases, and
tracing of these requires some combinatorial bookkeeping. That
bookkeeping is organized by \emph{set partitions} of $\{1,\dots,k\}$,
defined precisely in the proof below, with a worked example
for $k=2$ and $k=3$ included along the way.

\begin{theorem}
\label{thm:general}
For a given DMS and a given V2V length code with a dictionary $\mathcal{D}$
and code length function $\ell(\cdot)$, the $k$-th moment of the compression ratio of $n$
phrases is given by
\begin{equation}
\E\{R_n^k\} = \frac{1}{(k-1)!}\int_0^\infty t^{k-1}
\sum_{\tau\in P_k} \frac{n!}{(n-|\tau|)!}[\mu_0(t)]^{n-|\tau|}
\prod_{B\in\tau}\mu_{|B|}(t)\dif t,
\label{eq:moments}
\end{equation}
where $P_k$ denotes the collection of all set partitions of $\{1,\dots,k\}$;
$|\tau|$ denotes the number of subsets associated with a partition $\tau\in P_k$;
and $|B|$ denotes the cardinality of a subset $B$.
\end{theorem}

\begin{proof}
The integral representation of the function $\frac{1}{s^k}$ as a Laplace
transform is given by
\begin{equation}
\label{eq:Gamma}
\frac{1}{s^k} = \frac{1}{(k-1)!}\int_0^\infty t^{k-1}e^{-st}\dif t,~~~~s>0.
\end{equation}
Applying this with $s=\Sigma_n$ (which is a positive integer a.s.):
\begin{equation}
\label{eq:Rnk_int}
R_n^k = \frac{\Lambda_n^k}{\Sigma_n^k}
= \frac{\Lambda_n^k}{(k-1)!}\cdot\int_0^\infty t^{k-1}e^{-t\Sigma_n}\dif t.
\end{equation}
Taking expectations (again justified by finite-sum linearity, as above), we have
\begin{equation}
\label{eq:ERnk_int}
\E\{R_n^k\} = \frac{1}{(k-1)!}\int_0^\infty t^{k-1}\,\E\{\Lambda_n^k e^{-t\Sigma_n}\}\,\dif t.
\end{equation}
It remains to evaluate $\E\{\Lambda_n^k e^{-t\Sigma_n}\}$.
Now,
\begin{equation}
\label{eq:Lambda_expand}
\Lambda_n^k = \Bigl[\sum_{i=1}^n \ell(Y_i)\Bigr]^k
= \sum_{(i_1,\dots,i_k)\in\{1,\dots,n\}^k} \prod_{j=1}^k\ell(Y_{i_j}).
\end{equation}
Each index tuple $(i_1,\dots,i_k)$ induces a set partition $\tau$
of the set $\{1,\dots,k\}$ to a number of subsets. The slots $a$ and $b$ in
$\{1,\dots,k\}$ belong to the same
subset pertaining to partition $\tau$ if and only if $i_a=i_b$. 
We denote by $P_k$ for the collection of all set
partitions $\{\tau\}$ of $\{1,\dots,k\}$ (a finite set; $|P_1|=1$, $|P_2|=2$,
$|P_3|=5$), $|\tau|$ for the number of blocks of $\tau\in P_k$, and $|B|$
for the size of a block $B\in\tau$.\footnote{For example,
if $k=3$ and the tuple $(i_1,i_2,i_3)=(5,7,5)$, slots
$a=1$ and $b=3$ share the phrase index $5$ while slot $c=2$ has the different
index $7$, so the induced partition is $\tau=\{\{1,3\},\{2\}\}$, with
$|\tau|=2$ subsets of sizes $2$ and $1$. The other tuples inducing this
same $\tau$ are exactly those of the form $(j,j',j)$ for any two
\emph{distinct} phrase indices $j\ne j'$ in $\{1,\dots,n\}$.}
Because the phrase pairs $(L(Y_j),\ell(Y_j))_{j\ge1}$ are i.i.d., the
value of $\E\{\ell(Y_{i_1})\cdots\ell(Y_{i_k})e^{-t\Sigma_n}\}$ depends
on the tuple $(i_1,\dots,i_k)$ only through its induced partition
$\tau$, not on the specific phrase indices that realize it. This is what
lets us group the $n^k$ terms of \eqref{eq:Lambda_expand} by partition
type, rather than evaluating each term separately.
This index-partitioning technique -- grouping terms of a power of a
sum by which indices coincide, and exploiting exchangeability so that
each group's contribution depends only on the partition it induces --
is classical, in the spirit of the partition-lattice approach to
moments and cumulants of sums of random variables
\cite{speed-partitions}. What is needed beyond that classical identity
is that the exponential weight $e^{-t\Sigma_n}$ touches all $n$
phrases, not only the $k$ appearing in $\Lambda_n^k$: the $n-|\tau|$
phrases outside the partition's subsets still each contribute a factor
(through $\mu_0(t)$, below), and it is this joint accounting -- the
falling-factorial count together with the untouched phrases'
contribution -- that extends the classical moment-of-a-sum computation
to the joint quantity $\E\{\Lambda_n^k e^{-t\Sigma_n}\}$ actually needed
here.

Fix $\tau\in P_k$ with blocks $B_1,\dots,B_{|\tau|}$. A tuple
$(i_1,\dots,i_k)$ induces exactly this $\tau$ if and only if it assigns
the \emph{same} phrase index to every slot within a given subset, and
\emph{different} phrase indices to slots in different subsets.
So realizing $\tau$ amounts to choosing an
ordered assignment of $|\tau|$ \emph{distinct} phrase indices
$j_1,\dots,j_{|\tau|}\in\{1,\dots,n\}$, one per subset -- an injective
function from the $|\tau|$ blocks to the $n$ phrases. Assigning these
indices one block at a time, $B_1$ may receive any of the $n$ phrase
indices, $B_2$ any of the remaining $n-1$ (it must differ from $B_1$'s),
$B_3$ any of the remaining $n-2$, and so on, until $B_{|\tau|}$, which
has $n-|\tau|+1$ choices left; the number of
such assignments is the falling factorial
$n(n-1)\cdots(n-|\tau|+1)=\frac{n!}{(n-|\tau|)!}$.

With blocks $B_1,\dots,B_{|\tau|}$ assigned distinct phrase indices
$j_1,\dots,j_{|\tau|}$ respectively, the product
$\ell(Y_{i_1})\cdots\ell(Y_{i_k})$ collapses to
$\prod_{r=1}^{|\tau|}[\ell(Y_{j_r})]^{|B_r|}$, since subset $B_r$
contributes $|B_r|$ copies of the same factor $\ell(Y_{j_r})$. Splitting
$e^{-t\Sigma_n}=\prod_{j=1}^n e^{-tL(Y_j)}$ into the $|\tau|$ marked
phrases $j_1,\dots,j_{|\tau|}$ and the remaining $n-|\tau|$ phrases, and
using independence across phrases:
\begin{align}
\label{eq:factorize}
\E\Bigl\{\prod_{a=1}^k \ell(Y_{i_a})\cdot\prod_{j=1}^n e^{-tL(Y_j)}\Bigr\}
&= \prod_{r=1}^{|\tau|}\E\{[\ell(Y_{j_r})]^{|B_r|}e^{-tL(Y_{j_r})}\}\cdot
  \prod_{j\notin\{j_1,\dots,j_{|\tau|}\}}\E\{e^{-tL(Y_j)}\} \nonumber\\
&= \prod_{B\in\tau}\mu_{|B|}(t)\cdot[\mu_0(t)]^{n-|\tau|},
\end{align}
using the i.i.d.\ property and the definitions of $\mu_j(t)$ and
$\mu_0(t)$. This value depends only on $\tau$ (through its block sizes),
not on the specific phrases $j_1,\dots,j_{|\tau|}$ chosen to realize it
-- consistent with the partition-dependence noted above.

Every $\tau\in P_k$ contributes the same per-tuple expectation, as just
shown, so summing over all $\frac{n!}{(n-|\tau|)!}$ tuples that induce
it (the falling-factorial count derived above) and then over all
$\tau\in P_k$ gives
\begin{equation}
\label{eq:ELambdak}
\E\{\Lambda_n^ke^{-t\Sigma_n}\}
= \sum_{\tau\in P_k}\frac{n!}{(n-|\tau|)!}[\mu_0(t)]^{n-|\tau|}\prod_{B\in\tau}\mu_{|B|}(t).
\end{equation}
Substituting into \eqref{eq:ERnk_int} gives \eqref{eq:moments}.
\end{proof}

For $k=2$ and $k=3$,
the integrals are explicitly:
\begin{equation}
\label{eq:secondmoment}
\E\{R_n^2\} = \int_0^\infty t\bigl(n\mu_2(t)[\mu_0(t)]^{n-1}+n(n-1)[\mu_1(t)]^2[\mu_0(t)]^{n-2}\bigr)\dif t,
\end{equation}
\begin{align}
\label{eq:thirdmoment}
\E\{R_n^3\} = \frac{1}{2}\int_0^\infty t^2\bigl(&n\mu_3(t)[\mu_0(t)]^{n-1}
+3n(n-1)\mu_2(t)\mu_1(t)[\mu_0(t)]^{n-2}\\
&+n(n-1)(n-2)[\mu_1(t)]^3[\mu_0(t)]^{n-3}\bigr)\dif t.
\end{align}
From these three moments one obtains the mean, variance, and skewness
of $R_n$ via $\Var\{R_n\}=\E\{R_n^2\}-(\E\{R_n\})^2$ and
$\gamma_1=(\E\{R_n^3\}-3\E\{R_n\}\E\{R_n^2\}+2(\E\{R_n\})^3)/(\Var\{R_n\})^{3/2}$.

\subsection{Validation}

\label{sec:validation}
We now validate \eqref{eq:moments} on a simple V2V length code:
a binary memoryless source
with $P(0)=1-P(1)=p=0.8$, parsed by the Tunstall tree with dictionary
$\{00,01,1\}$, and coded with the Huffman assignment
$(\ell(00),\ell(01),\ell(1))=(1,2,2)$. The same source, tree, and
codeword assignment are used again for the bias-decomposition example
in Section~\ref{sec:comparison}. This is
a complete prefix code:
$2^{-1}+2^{-2}+2^{-2}=1$, so Kraft's inequality holds with equality, and
every codeword length is a positive integer. The phrase
probabilities are $Q(00)=p^2=0.64$, $Q(01)=p(1-p)=0.16$, $Q(1)=1-p=0.2$,
giving $\bar L=1.8$, $\bar\ell=1.36$, and $\rho=\bar\ell/\bar L=0.7556$.

For this model, $\mu_j(t)=\E\{[\ell(Y)]^je^{-tL(Y)}\}$ is the explicit
three-term sum
\begin{equation}
\mu_j(t) = 0.64\cdot1^j e^{-2t} + 0.16\cdot2^j e^{-2t} + 0.2\cdot2^j e^{-t},
\end{equation}
so, e.g., $\mu_0(t)=0.8e^{-2t}+0.2e^{-t}$ and $\mu_1(t)=0.96e^{-2t}+0.4e^{-t}$.
The three integrals \eqref{eq:mean}, \eqref{eq:secondmoment}, \eqref{eq:thirdmoment} are
evaluated by standard numerical quadrature (Gaussian quadrature
after the substitution $t=u/n$ to handle the concentration near $t=0$
for larger $n$).

At $n=50$ phrases, the exact formulas give:
\begin{equation}
\label{eq:exact50}
\E\{R_{50}\}=0.7571,\quad\Var\{R_{50}\}=0.003230,\quad
\gamma_1=0.2878.
\end{equation}
For validation, we performed a Monte Carlo simulation of $1{,}000{,}000$
independent realizations of $50$ phrase pairs $(L(Y_i),\ell(Y_i))$, drawn
according to the three phrase probabilities above, computing the
realized ratio $R_{50}=\Lambda_{50}/\Sigma_{50}$ for each. The empirical
estimates are:
\begin{equation}
\label{eq:mc50}
\widehat{\E\{R_{50}\}}=0.7571,\quad
\widehat{\Var\{R_{50}\}}=0.003226,\quad
\widehat{\gamma_1}=0.2875.
\end{equation}
Agreement is to three or four significant figures throughout -- the
skewness $\gamma_1$ in particular is the first quantity that
\eqref{eq:thirdmoment} provides beyond what a first-order
approximation would give -- and the exact formulas
require no simulation effort at all.

\subsection{The Edgeworth approximation}

\label{sec:edgeworth}
The three exact moment formulas \eqref{eq:mean}--\eqref{eq:thirdmoment} can be used to
construct a sharper approximation to the CDF of $R_n$ than the ordinary
CLT (Gaussian) approximation.
By the bivariate CLT, the pair $(\Lambda_n-n\bar\ell,\Sigma_n-n\bar L)/\sqrt{n}$ converges
in distribution to a bivariate normal as $n\to\infty$.
By the delta method \cite[Theorem~3.1]{vandervaart}, any function
that is continuously differentiable at the limiting mean again yields
a normal limit, with asymptotic variance determined by the gradient of
the function and the full limiting covariance matrix.
Applied to the map $(\lambda,m)\mapsto\lambda/m$, which is
continuously differentiable at $(\bar\ell,\bar L)$ since $\bar L>0$, this
implies that
$Z_n\dfn(R_n-\E\{R_n\})/\sqrt{\Var\{R_n\}}$
converges in distribution to a standard normal $\mathcal{N}(0,1)$.
The CLT approximation to the CDF $F(x)\dfn\Pr\{R_n\le x\}$ is therefore
$\Phi_{\rm nor}(z)$, where $z=(x-\E\{R_n\})/\sqrt{\Var\{R_n\}}$ and $\Phi_{\rm nor}$ is the
standard normal CDF. But this approximation ignores the skewness of the
distribution, which is $O(1/\sqrt{n})$ and non-negligible at moderate $n$.

The one-term Edgeworth expansion corrects for the skewness. It approximates
$F$ by 
\begin{equation}
F_{\mbox{\tiny Edgeworth}}(x)\dfn \Phi_{\rm
nor}(z)-\varphi(z)\,\frac{\gamma_1}{6}(z^2-1),
\label{eq:edgeworth}
\end{equation}
where $z=(x-\E\{R_n\})/\sqrt{\Var\{R_n\}}$, $\varphi=\Phi'_{\rm nor}$ is the
standard normal density, and $\gamma_1$ is the skewness of $R_n$.
The approximation \eqref{eq:edgeworth} is a standard result in the
theory of Edgeworth expansions for sums of i.i.d.\ random variables
\cite[Ch.~XVI]{feller}; its validity for the ratio $R_n=\Lambda_n/\Sigma_n$
follows from the delta method applied to the bivariate CLT for
$(\Lambda_n,\Sigma_n)$, under standard moment conditions on the phrase pair
$(L(Y),\ell(Y))$.
The key point is that all three parameters entering \eqref{eq:edgeworth}
-- $\E\{R_n\}$, $\Var\{R_n\}$, and $\gamma_1$ -- are available
from \eqref{eq:mean}--\eqref{eq:thirdmoment} for every finite $n$. 

Continuing the same numerical example as in Subsection \ref{sec:validation},
Table \ref{tab:edgeworth} shows the absolute approximation error
$|F(x)-F_{\mbox{\tiny Edgeworth}}(x)|$ for $n=50$, at several standardized
values $z$, comparing the CLT approximation against
\eqref{eq:edgeworth} with the exact $\gamma_1=0.2878$:

\begin{table}[h]
\centering
\begin{tabular}{cccc}
\hline
$z$ & $x$ & CLT error & Edgeworth error\\
\hline
$-2.0$ & $0.6434$ & $0.00773$ & $0.00004$\\
$-0.5$ & $0.7287$ & $0.01881$ & $0.00614$\\
$\phantom{-}0.0$ & $0.7571$ & $0.01972$ & $0.00058$\\
$\phantom{-}0.5$ & $0.7855$ & $0.01111$ & $0.00155$\\
$\phantom{-}2.0$ & $0.8708$ & $0.00574$ & $0.00203$\\
\hline
\end{tabular}
\caption{Absolute CDF approximation error at $n=50$, against empirical
CDF from $2{,}000{,}000$ Monte Carlo trials.}
\label{tab:edgeworth}
\end{table}

The Edgeworth correction reduces the approximation error at every
tabulated point, by a factor ranging from about $3$ to nearly $200$
depending on $z$; averaged over the range shown, the mean absolute
error drops by a factor of about $3.4$ and the maximum error by a
factor of about $2.7$. This is a more modest improvement than a smooth,
continuously-supported distribution would show (compare
Remark~\ref{rem:lattice} below), since $L(Y)$ and $\ell(Y)$ each take
only two values here, but it is a genuine improvement for an
actual code.

Figure~\ref{fig:edgeworth} shows this comparison as a continuous
function of $z$ rather than at the five tabulated points, against a
$20{,}000{,}000$-trial Monte Carlo baseline. The CLT error traces out a
broad envelope peaking near $z=0$, while the Edgeworth error stays
consistently lower across the entire range shown. Both curves show
some oscillation rather than a perfectly smooth profile; this is a
genuine finite-alphabet effect discussed in
Remark~\ref{rem:lattice} below.

\begin{figure}[h]
\centering
\includegraphics[width=0.75\textwidth]{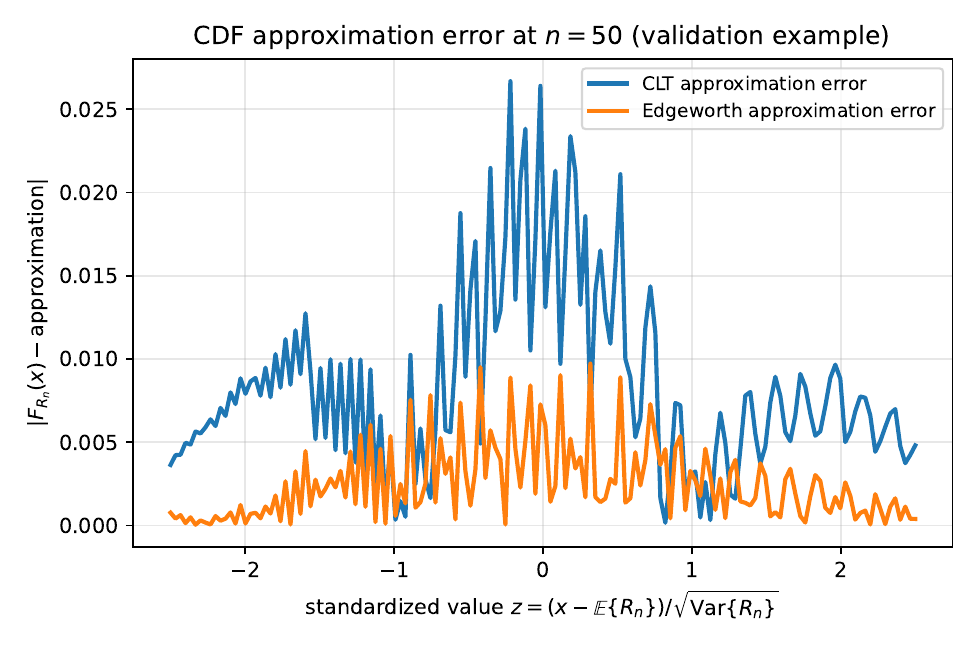}
\caption{Absolute CDF approximation error, CLT vs.\ Edgeworth
\eqref{eq:edgeworth}, as a function of the standardized value $z$, for
the validation example of Section~\ref{sec:validation} at $n=50$.
Ground truth is a $20{,}000{,}000$-trial Monte Carlo simulation.}
\label{fig:edgeworth}
\end{figure}

\begin{remark}[A finite-alphabet effect]
\label{rem:lattice}
Because $L(Y)\in\{1,2\}$ and $\ell(Y)\in\{1,2\}$ each take only two
values, $\Lambda_n$ and $\Sigma_n$ are both sums of small integers, so
$R_n=\Lambda_n/\Sigma_n$ takes values in a finite, coarsely-spaced set of
rationals for any fixed $n$ -- its exact CDF is a step function with
comparatively few, comparatively large jumps. Neither the CLT nor the
Edgeworth approximation is built to track a step function's exact
jumps (both are continuous-$x$ approximations, designed for the
lattice spacing to vanish as $n\to\infty$), so both error curves in
Figure~\ref{fig:edgeworth} inherit some oscillation from this
discreteness, on top of their genuine approximation error. This is a
real feature of finite, small-dictionary V2V codes, not an artifact:
a synthetic phrase-length law with a smooth, unbounded, or
finely-spaced support -- as used in an earlier version of this
example -- would mask it entirely. The Edgeworth correction still
reduces the error at every point tested, as Table~\ref{tab:edgeworth}
shows; it is simply a less clean improvement than the idealized,
non-lattice case.
\end{remark}

\subsection{Asymptotic Approximation via the Laplace method of integration}

\label{sec:laplace}
For large $n$, the integrands in \eqref{eq:moments} are sharply concentrated
near $t=0$ and can be evaluated asymptotically by the Laplace method of
integration. The change of integration variabe
$t=u/n$ transforms the integrals into a form where the
concentration is explicit and a systematic expansion in $1/n$ can be
carried out.
There are two versions of Laplace's method for an integral
$\int_a^b g(t)e^{-nf(t)}\dif t$, depending on where $f$ attains its minimum
over the domain of integration. If the minimum occurs at an
\emph{interior} point $t_0\in(a,b)$ with $f'(t_0)=0$, the local behavior of $f$
near $t_0$ is quadratic, and so, $e^{-nf(t)}$ is locally approximated by a
Gaussian, $e^{-nf(t_0)}\cdot\exp\left\{-\frac{nf''(t_0)}{2}(t-t_0)^2\right\}$,
giving the familiar $\sqrt{\frac{2\pi}{nf''(t_0)}}$ prefactor. If
instead the minimum occurs at a \emph{boundary} point of the domain ($a$ or
$b$) with $f'(t_0)\ne0$, the local behavior of $f$ is linear rather than
quadratic, and $e^{-nf(t)}$ is locally approximated at the vicinity of $t_0$ by a plain
exponential, $e^{-nf(t_0)}e^{-nf'(t_0)(t-t_0)}$, rather than a
Gaussian; its moments against a power series in $(t-t_0)$ are then
elementary Gamma-function integrals rather than Gaussian moments (see,
e.g., \cite[Sec.~4.3]{debruijn} or \cite[p.~48]{merhav-weinberger}).
This second, boundary case is the one relevant here: as shown below,
$f(t)\dfn-\ln\mu_0(t)$ attains its minimum over $t\in[0,\infty)$ at the
boundary point $t_0=0$ with $f'(0)=\bar L\ne0$, so $[\mu_0(t)]^{n}=e^{-nf(t)}$
is approximated near $t=0$ by the exponential $e^{-n\bar Lt}$.
The derivation below carries this out to one further
order in $1/n$ than the leading term alone provides, since the bias
coefficient $C$ requires the $O(1/n)$ correction.

Let $f(t)=-\ln\mu_0(t)$ denote the negative cumulant generating function (CGF) of
$-L(Y)$. Since $L(Y)\ge1$ a.s., $f(0)=0$ and $f'(0)=\E\{L(Y)\}=\bar L>0$.
Furthermore, $f$ is strictly concave as $f''(t)<0$ for all $t\ge0$.
In particular,
$f(t)>0$ for all $t>0$, so
$[\mu_0(t)]^{n-1}=e^{-(n-1)f(t)}\to0$
exponentially fast in $n$ for every fixed $t>0$. The integrand of
\eqref{eq:mean}, which is proportional to $\mu_1(t)[\mu_0(t)]^{n-1}$, therefore
concentrates near $t_0=0$ as $n\to\infty$. Substituting $t=u/n$ gives
$[\mu_0(u/n)]^{n-1}=e^{-(n-1)f(u/n)}\approx e^{-\bar L u}$ (to leading
order for large $n$), and after the substitution the integral becomes
$\int_0^\infty\mu_1(u/n)e^{-(n-1)f(u/n)}\dif u$, which is $O(1)$.

\paragraph{The bias term.}
Write $\bar L=\E\{L(Y)\}$, $\kappa_2=\Var\{L(Y)\}$, $\bar\ell=\E\{\ell(Y)\}$,
$c_{\ell L}=\E\{\ell(Y)L(Y)\}$. The $O(1/n)$ expansion of $\E\{R_n\}$
follows from a general fact about this boundary case of the Laplace
method, stated here and proved in Appendix~A (see also
\cite[Sec.~4.3]{debruijn} and \cite[p.~48]{merhav-weinberger} for the
general theory).

\begin{lemma}[Boundary Laplace expansion]
\label{lem:boundary_laplace}
Let $f:[0,\infty)\to[0,\infty)$ satisfy $f(0)=0$, be twice
continuously differentiable at $t=0$ with $f'(0)=a$ for some $a>0$
and $f''(0)=-b$, and let $h$ be continuously differentiable at $t=0$
with $h(0)=h_0$, $h'(0)=h_1$ (both suitably regular for the expansion
below to hold, as verified for $f=-\ln\mu_0$ above). Then, for any
integers $k\ge1$ and $m\ge1$, as $n\to\infty$,
\begin{eqnarray}
\label{eq:boundary_laplace}
& &\frac{n!}{(n-m)!}\int_0^\infty t^{k-1}h(t)e^{-(n-m)f(t)}\dif t\nonumber\\
&=&n^{m-k}\biggl[\frac{(k-1)!h_0}{a^k}+
\frac{(k-1)!}{2na^{k+2}}\Bigl(a^2h_0m(2k-m+1)+2ah_1k+bh_0k(k+1)\Bigr)
+O\left(\frac{1}{n^{2}}\right)\biggr]. \notag
\end{eqnarray}
\end{lemma}

The case $k=m=1$, $h=g$ reduces to
$$n\cdot\int_0^\infty g(t)e^{-(n-1)f(t)}\dif t = 
\frac{g_0}{a}+\frac{a^2g_0+ag_1+bg_0}{na^3}+O\left(\frac{1}{n^{2}}\right),$$
the form used just below.

Applying Lemma~\ref{lem:boundary_laplace} with $f(t)=-\ln\mu_0(t)$
and $g=\mu_1$: from the previous paragraph, $a=\bar L$ and $b=\kappa_2$;
and since $\mu_1(0)=\E\{\ell(Y)\}=\bar\ell$ and
\begin{equation}
\mu_1'(t)=\frac{\dif}{\dif t}\E\{\ell(Y)e^{-tL(Y)}\}=-\E\{L(Y)\ell(Y)e^{-tL(Y)}\},
\end{equation}
so $\mu_1'(0)=-c_{\ell L}$, we have $g_0=\bar\ell$ and $g_1=-c_{\ell L}$.
Substituting into \eqref{eq:boundary_laplace},
\begin{equation}
\E\{R_n\} = n\cdot\int_0^\infty\mu_1(t)[\mu_0(t)]^{n-1}\dif t
= \frac{\bar\ell}{\bar L}+\frac1n\cdot\frac{\bar L^2\bar\ell-\bar Lc_{\ell L}+
\kappa_2\bar\ell}{\bar L^3}+O\left(\frac{1}{n^2}\right)
= \rho+\frac{C}{n}+O\left(\frac{1}{n^2}\right)
\end{equation}
with
\begin{equation}
C = \frac{\bar\ell(\bar L^2+\kappa_2)-c_{\ell L}\bar L}{\bar L^3}.
\end{equation}
We state this as Proposition~\ref{prop:bias}
below.

\begin{proposition}[Asymptotic bias]
\label{prop:bias}
\begin{equation}
\label{eq:bias_formula}
\lim_{n\to\infty}n\cdot(\E\{R_n\}-\rho)
=\frac{\bar\ell(\bar L^2+\kappa_2)-c_{\ell L}\bar L}{\bar L^3}\dfn C,
\end{equation}
where $\rho=\bar\ell/\bar L$.
\end{proposition}

A comment is in order concerning an alternative
route for calculating the bias -- the delta-method.
Since $\bar U_n:=\Sigma_n/n$ and $\bar V_n:=\Lambda_n/n$ are exact
sample means of i.i.d.\ data, $\E\{\bar U_n\}=\bar L$ and
$\E\{\bar V_n\}=\bar\ell$ exactly, for every $n$. Writing
$g(u,v):=v/u$ and Taylor-expanding $R_n=g(\bar U_n,\bar V_n)$ to
second order around $(\bar L,\bar\ell)$, the first-derivative terms
vanish upon taking expectations, since the sample means are exactly
unbiased -- so the entire bias comes from the second derivatives:
\begin{equation}
C = \frac12\Bigl(g_{uu}\Var\{L\}+2g_{uv}\Cov\{L,\ell\}\Bigr)
= \frac{\rho\,\Var\{L\}-\Cov\{L,\ell\}}{\bar L^2},
\end{equation}
using $g_{uu}=2\rho/\bar L^2$, $g_{uv}=-1/\bar L^2$, $g_{vv}=0$ at
$(\bar L,\bar\ell)$ -- algebraically identical to
\eqref{eq:bias_formula}, with $\Var\{L\}=\kappa_2$ and
$\Cov\{L,\ell\}=c_{\ell L}-\bar L\bar\ell$.

This is a useful independent cross-check, but not a substitute for the
derivation above. The delta method, by construction, gives only the
\emph{leading} $O(1/n)$ correction: it starts from an asymptotic
expansion of $R_n$ itself and stops at second order once that
correction is in hand, since the first-derivative terms vanish
identically and nothing forces the expansion any further. The
Laplace-method route above instead starts from the \emph{exact}
formula \eqref{eq:mean}, valid at every finite $n$, and reaches the
bias by expanding \emph{that} formula asymptotically -- so the same
machinery, carried one order further, would give the $O(1/n^2)$ term
too, and (via Theorem~\ref{thm:general} and
Lemma~\ref{lem:boundary_laplace} applied repeatedly, as just below)
gives every higher moment -- variance, skewness, and beyond -- from a
single combinatorial formula, rather than requiring a fresh,
increasingly delicate multivariate Taylor expansion for each one in
turn. It also extends uniformly to the Markov case of
Section~\ref{sec:markov}: there, the delta method's natural
generalization requires tracking how a phrase's own fluctuation
correlates with the state it leaves the chain in -- a genuine
complication with no analogue here -- whereas an eigenvalue-perturbation
argument handles this automatically,
by perturbing the joint spectral object directly rather than
decomposing the bias into pieces by hand.

\paragraph{The variance term.}
Lemma~\ref{lem:boundary_laplace}, applied twice more, yields the
asymptotic variance directly. Write $\E\{R_n^2\}=I_1+I_2$ where
\begin{equation}
\label{eq:I1I2}
I_1=\int_0^\infty tn\mu_2(t)[\mu_0(t)]^{n-1}\dif t,\quad
I_2=\int_0^\infty tn(n-1)[\mu_1(t)]^2[\mu_0(t)]^{n-2}\dif t,
\end{equation}
each matching the Lemma's left-hand side with $k=2$ (since the extra
factor is $t=t^{k-1}$).
For $I_1$, $m=1$ and $h=\mu_2$, so $h_0=\mu_2(0)=s_{\ell2}$ where
$s_{\ell2}:=\E\{[\ell(Y)]^2\}$; since $I_1$ itself turns out to
contribute only at order $1/n$ to $\E\{R_n^2\}$, only the Lemma's
leading term is needed:
\begin{equation}
I_1 = \frac{s_{\ell2}}{n\bar L^2}+O\left(\frac{1}{n^2}\right).
\end{equation}
For $I_2$, $m=2$ and $h=[\mu_1]^2$, so $h_0=[\mu_1(0)]^2=\bar\ell^2$
and, by the product rule, $h_1=2\mu_1(0)\mu_1'(0)=-2\bar\ell\,c_{\ell L}$;
since $I_2$ contributes at leading order $O(1)$, both terms of the
Lemma are needed:
\begin{equation}
I_2 = \frac{\bar\ell^2}{\bar L^2}
+\frac{1}{n}\cdot\frac{\bar\ell\bigl(3\bar L^2\bar\ell-4\bar L\,c_{\ell L}+3\kappa_2\bar\ell\bigr)}{\bar L^4}
+O\left(\frac{1}{n^2}\right).
\end{equation}

\begin{proposition}[Asymptotic variance]
\label{prop:var}
With $s_{\ell2}=\E\{[\ell(Y)]^2\}$:
\begin{equation}
\label{eq:var_formula}
\lim_{n\to\infty}n\cdot\Var\{R_n\}=
\frac{\bar L^2(s_{\ell2}-\bar\ell^2)-2\bar L\bar\ell(c_{\ell L}-\bar L\bar\ell)+\kappa_2\bar\ell^2}{\bar L^4}.
\end{equation}
\end{proposition}

\begin{proof}
Adding $I_1$ and $I_2$ above gives $\E\{R_n^2\}$ to $O(1/n)$; subtracting
$(\E\{R_n\})^2=\rho^2+2\rho C/n+O(n^{-2})$ (squaring
Proposition~\ref{prop:bias}) and multiplying by $n$, the $O(1)$ terms
cancel exactly (since $\bar\ell^2/\bar L^2=\rho^2$), leaving
\eqref{eq:var_formula} after collecting the remaining terms over the
common denominator $\bar L^4$. 
\end{proof}

On the validation example of Section~\ref{sec:validation}, the analytic
formulas of Propositions \ref{prop:bias} and \ref{prop:var} evaluate to
a bias coefficient of $0.076818$ and an asymptotic scaled variance of
$0.159000$. The exact formula \eqref{eq:moments} evaluated at large $n$ gives
$n(\E\{R_n\}-\rho)\to0.076818$ and $n\Var\{R_n\}\to0.159000$, matching to
six significant figures and confirming both propositions numerically.

\begin{remark}[Design implication]
Proposition~\ref{prop:bias} shows that for a fixed parsing tree (hence
fixed $\bar L$ and $\Var\{L\}$), the bias $\E\{R_n\}-\rho$ is minimized
(i.e., $\E\{R_n\}$ converges to $\rho$ fastest) when $\Cov\{L,\ell\}$ is
maximized. This has a clear design interpretation: longer phrases should
be assigned longer codewords, i.e., $\Cov\{L,\ell\}$ should be made
positive. Whether a specific codeword assignment achieves this depends
on the source and tree; the bias formula gives a precise, finite-$n$
quantification of how far any given code deviates from the optimum.
\end{remark}

\subsection{Large deviations of the compression ratio}
\label{sec:largedev}

The Edgeworth expansion of Section~\ref{sec:edgeworth} characterizes the
behavior of $R_n$ in the $O(1/\sqrt{n})$ fluctuation region around $\rho$.
For deviations of fixed size $\Delta>0$ (independent of $n$), the
probability $\Pr\{R_n > \rho+\Delta\}$ decays exponentially in $n$, and
its exact exponential rate follows directly from Cram\'er's theorem.

The key observation is that the event $\{R_n > \rho+\Delta\}$ is
equivalent to a standard large-deviations event for a sum of i.i.d.\
random variables. Since
\begin{equation}
\label{eq:ld_equiv}
R_n > \rho+\Delta
\;\iff\;
\frac{\Lambda_n}{\Sigma_n} > \rho+\Delta
\;\iff\;
\sum_{i=1}^n \bigl[\ell(Y_i) - (\rho+\Delta)L(Y_i)\bigr] > 0,
\end{equation}
the event is determined by whether the partial sum of the i.i.d.\ random
variables $Z_i^{(\Delta)} := \ell(Y_i) - (\rho+\Delta)L(Y_i)$ exceeds zero.
Since $\E\{Z_i^{(\Delta)}\} = \E\{\ell(Y)\} - (\rho+\Delta)\E\{L(Y)\} = -\Delta\E\{L(Y)\} < 0$
for $\Delta>0$, this is a large-deviations event with the sum crossing
zero against its drift. An identical reduction handles the lower tail:
$\{R_n < \rho-\Delta\} \iff \sum_i [\ell(Y_i)-(\rho-\Delta)L(Y_i)] < 0$, with
$\E\{\ell(Y)-(\rho-\Delta)L(Y)\} = \Delta\E\{L(Y)\} > 0$.

By Cram\'er's theorem, for every $\Delta>0$ (all logarithms in this
subsection are natural logarithms):
\begin{align}
\frac{1}{n}\log\Pr\{R_n > \rho+\Delta\} &\;\to\; -I_+(\Delta), \label{eq:ld_upper}\\
\frac{1}{n}\log\Pr\{R_n < \rho-\Delta\} &\;\to\; -I_-(\Delta), \label{eq:ld_lower}
\end{align}
where the rate functions are
\begin{equation}
\label{eq:ld_rate_upper}
I_+(\Delta) \;=\; \sup_{\theta > 0}
\Bigl\{-\log\E\bigl\{e^{\theta[\ell-(\rho+\Delta)L]}\bigr\}\Bigr\},
\end{equation}
\begin{equation}
\label{eq:ld_rate_lower}
I_-(\Delta) \;=\; \sup_{\theta < 0}
\Bigl\{-\log\E\bigl\{e^{\theta[\ell-(\rho-\Delta)L]}\bigr\}\Bigr\}.
\end{equation}
For a real parameter $\theta$ and constant $c$, let $P_\theta^{(c)}$
denote the exponentially tilted law of $(L,\ell)$ obtained by weighting
each realization proportionally to $e^{\theta(\ell-cL)}$, and write
$\E_\theta^{(c)}\{\cdot\}$ for the corresponding expectation; below, $c$
is $\rho+\Delta$ for the upper tail and $\rho-\Delta$ for the lower
tail, and we abbreviate $P_\theta^{(c)}$ and $\E_\theta^{(c)}$ by
$P_\theta$ and $\E_\theta$ since $c$ is clear from context.
The supremum over $\theta>0$ in $I_+$ is unconstrained (the CGF is finite
for all $\theta>0$ provided $\ell$ has all exponential moments, which holds
for a finite alphabet); for a finite dictionary $L$ is bounded, so the CGF of
$\ell-(\rho-\Delta)L$ is finite for all $\theta\in\mathbb{R}$, and the
supremum over $\theta<0$ in $I_-$ is unconstrained; for codes with
unbounded phrase lengths, the supremum may be constrained to an
interval $(\theta_{\min},0)$.
The supremum in each case is achieved at a unique $\theta^*$. For the
upper tail, $\theta^*>0$ is the unique positive root of
\begin{equation}
\label{eq:ld_tilt_upper}
\frac{\E\{[\ell-(\rho+\Delta)L]\,e^{\theta^*[\ell-(\rho+\Delta)L]}\}}
{\E\{e^{\theta^*[\ell-(\rho+\Delta)L]}\}} = 0,
\end{equation}
and for the lower tail, $\theta^*<0$ is the unique negative root of
\begin{equation}
\label{eq:ld_tilt_lower}
\frac{\E\{[\ell-(\rho-\Delta)L]\,e^{\theta^*[\ell-(\rho-\Delta)L]}\}}
{\E\{e^{\theta^*[\ell-(\rho-\Delta)L]}\}} = 0.
\end{equation}
In both cases the condition expresses that the compression ratio under
the tilted distribution $P_{\theta^*}$ equals the target level:
$\E_{\theta^*}\{\ell\}/\E_{\theta^*}\{L\}=\rho+\Delta$ (upper tail) or
$\rho-\Delta$ (lower tail). Thus $P_{\theta^*}$ is the unique exponential
tilt of the original phrase-pair law under which the large-deviations
event is typical.

The rate functions $I_+,I_-$ are not built from a separate toolkit:
$\E\{e^{\theta[\ell-(\rho+\Delta)L]}\}=\E\{e^{\theta\ell-tL}\}$ along
the ray $t=\theta(\rho+\Delta)$ (and similarly for $I_-$ along
$t=\theta(\rho-\Delta)$), and this joint exponential moment is exactly
the object that reduces to $\mu_0(t)$ at $\theta=0$ and to $\mu_1(t)$
upon differentiating in $\theta$ at $\theta=0$
(Section~\ref{sec:background}). The exact moments of
Sections~\ref{sec:moments}--\ref{sec:laplace} and the large-deviations
rate functions here are two different slices of this same underlying
two-parameter family, evaluated in different regimes: $t$ alone for
the moments, and along the rays above for the tails.

\section{Comparison with V2F and F2V length codes}
\label{sec:comparison}

The three code families differ in which of the phrase length $L$ and
codeword length $\ell$ they allow to vary: F2V fixes $L$, V2F fixes
$\ell$, and V2V lets both vary. The fact that V2V can match or beat
the other two on the asymptotic rate $\rho$ is not surprising -- it
strictly contains them as special cases. The more informative question
is what this extra freedom does to the \emph{finite-$n$} behavior
captured by $C$, and what, if anything, it reveals about why some V2V
constructions substantially outperform both alternatives while others
do not.

\paragraph{Setup.} Fix a binary memoryless source with $p=P(0)=0.8$
($H=0.722$ bits/symbol) and the Tunstall tree $\{00,01,1\}$
($\bar L=1.8$, $\Var\{L\}=0.16$).

\paragraph{The structural decomposition.}
The bias formula of Proposition~\ref{prop:bias} rewrites as
\begin{equation}
\label{eq:bias_structure}
C = \frac{\rho\,\Var\{L\} - \Cov\{L,\ell\}}{\bar L^2},
\end{equation}
which separates cleanly by code family:
\begin{itemize}
\item \textbf{F2V}: $L$ is deterministic, so $\Var\{L\}=\Cov\{L,\ell\}=0$
and $C=0$ identically -- not a design achievement, but a triviality of
having no phrase-length randomness to create bias from. It buys
nothing for $\rho$.
\item \textbf{V2F}: $\ell$ is constant, so $\Cov\{L,\ell\}=0$ always,
giving $C=\rho\Var\{L\}/\bar L^2>0$ with no freedom to reduce it: the
fixed codeword length that limits $\rho$ to $O(1/\bar L)$ redundancy is
the same constraint that locks in this bias.
\item \textbf{V2V}: both vary, so $\Cov\{L,\ell\}$ is a genuine design
variable. Correlating $L$ and $\ell$ positively -- longer phrases
getting longer codewords -- reduces $C$ below what V2F could achieve
with the same tree and the same $\Var\{L\}$.
\end{itemize}
This freedom cuts both ways, though: the Huffman codeword assignment
that \emph{minimizes} $\rho$ need not create positive covariance. For
the tree above, Huffman assigns the shortest codeword to the most
probable phrase ($00$, $P=0.64$), which happens to be one of the
\emph{longest} phrases ($L=2$), giving $\Cov\{L,\ell\}=-0.128<0$ and
$C_{\rm V2V}=0.077$ -- \emph{larger} than $C_{\rm V2F}=0.055$ despite
$\rho_{\rm V2V}=0.756<\rho_{\rm V2F}=1.111$. The same tree admits a
different codeword assignment (giving the shortest codeword to the
\emph{shortest} phrase instead) that achieves $\Cov\{L,\ell\}=+0.160$
and $C=0$ exactly, at the cost of a worse $\rho=1$. V2F and F2V have no
such choice to make; V2V's freedom is precisely the ability to trade
between these two objectives, in either direction.

\paragraph{Why this matters: the Khodak code.}
The tension above raises the question of whether $\rho$ and $C$ can be
improved \emph{together} rather than traded off, and the Khodak
code \cite{drmota-szpan,bugeaud08} shows that they can. By the
conservation of entropy \cite{savari-cons}, $H(Y)=\bar L\,H$ for any
parsing tree, where $H(Y)$ is the phrase entropy; the Khodak
construction assigns near-Shannon codewords, $\ell(y)\approx-\log_2Q(y)$,
to a tree deliberately built to preserve non-uniform leaf
probabilities. By the asymptotic equipartition property, a
\emph{typical} phrase of length $L(y)$ satisfies $Q(y)\approx2^{-L(y)H}$,
so $\ell(y)\approx L(y)H$ holds for the typical phrases that carry
essentially all the probability mass as $L(y)$ grows -- not for every
individual leaf (a phrase consisting entirely of the single most
probable symbol, for instance, gets a codeword far shorter than
$L(y)H$), but for enough of the distribution to drive the averages
below. This achieves redundancy $r:=\rho-H=O(\bar L^{-5/3})$ --
strictly better than the $O(1/\bar L)$ available to V2F or to
Tunstall--Huffman. The same typical-phrase approximation gives
$\Cov\{L,\ell\}\approx H\,\Var\{L\}$, so substituting into
\eqref{eq:bias_structure},
\begin{equation}
\label{eq:khodak_C}
C \;\approx\; \frac{\rho\Var\{L\}-H\Var\{L\}}{\bar L^2}
\;=\; \frac{(\rho-H)\Var\{L\}}{\bar L^2} \;=\; \frac{r\,\Var\{L\}}{\bar L^2}
\;=\; O(r) \;=\; O(\bar L^{-5/3}).
\end{equation}
The bias coefficient shrinks \emph{at the same rate as the redundancy}
-- not a second, independent achievement, but a direct consequence of
the same design principle. Compare V2F: $\ell$ cannot track $L$ at
all, so $\Cov\{L,\ell\}=0$ is permanent regardless of how the tree is
built, and neither the $O(1/\bar L)$ floor on $r$ nor the resulting
$C=\rho\Var\{L\}/\bar L^2$ can be improved by this mechanism.

For a binary Bernoulli$(2/3)$ source, the explicit
Bugeaud--Drmota--Szpankowski construction \cite{bugeaud08} (convergent
denominator $q=5$) gives $\bar L=116.4$, $\rho=0.9211$ against
$H=0.9183$ ($r=0.0028$), and $\Var\{L\}=4824$; \eqref{eq:khodak_C}
predicts $C\approx0.00097$, close to the exact value $C=0.000962$ from
Proposition~\ref{prop:bias} ($\Cov\{L,\ell\}=4430$ against the
predicted $H\Var\{L\}=4430$). Table~\ref{tab:khodak} confirms the
resulting $O(1/n)$ convergence numerically, using the exact moment
formula \eqref{eq:mean} at every $n$:
\begin{center}
\begin{tabular}{rccc}
\hline
$n$ & Exact $\E\{R_n\}$ & Asymp.\ $\rho+C/n$ & Error (approx.$-$exact)\\
\hline
$1$    & $0.91994$ & $0.92206$ & $+0.00213$\\
$2$    & $0.92269$ & $0.92158$ & $-0.00111$\\
$5$    & $0.92144$ & $0.92129$ & $-0.00015$\\
$10$   & $0.92123$ & $0.92120$ & $-0.000033$\\
$20$   & $0.92116$ & $0.92115$ & $-0.0000082$\\
$50$   & $0.92112$ & $0.92112$ & $-0.0000013$\\
$100$  & $0.92111$ & $0.92111$ & $\approx 0$\\
$\infty$ & $\rho=0.92110$ & $\rho=0.92110$ & $0$\\
\hline
\end{tabular}
\captionof{table}{Exact vs.\ first-order asymptotic $\E\{R_n\}$ for the
explicit Khodak code of \cite{bugeaud08}, illustrating the $O(1/n)$
convergence rate of Proposition~\ref{prop:bias}.}
\label{tab:khodak}
\end{center}
The approximation is accurate to five decimal places by $n=50$, well
before the asymptotic regime is actually reached.

\section{Extension to Markov sources}
\label{sec:markov}

This section repeats the memoryless program of
Section~\ref{sec:moments} with the i.i.d.\ assumption dropped: a
boundary state is identified that renders successive phrases Markov
rather than independent (\S\ref{sec:markov_model}); the scalar
quantities $\mu_j(t)$ become matrices indexed by that state
(\S\ref{sec:markov_transforms}); the mean formula of
Theorem~\ref{thm:mean} becomes a matrix product
(\S\ref{sec:markov_worked}, validated against direct simulation); and
the Laplace-method bias analysis of Section~\ref{sec:laplace} becomes
a matrix-eigenvalue perturbation, with the correction term now
characterized by a Poisson equation rather than a plain derivative
(\S\ref{sec:markov_bias}, with a pointer to the underlying argument in
Appendix~B. The memoryless case reappears
throughout as an exact special case, not a separate limit.

We now extend the source model to the Markov case.
Specifically, in this section, $X_1,X_2,\dots$ is assumed an irreducible,
time-homogeneous, \emph{first-order} Markov chain on the alphabet
$\mathcal{X}$ of cardinality $\alpha$,
with transition probabilities $\Pr\{X_{k+1}=x'\mid X_k=x\}$, and initial state
$X_0$. The parsing tree, dictionary,
and codeword assignment are exactly as before.

\subsection{Source model: the boundary state $S_i$}
\label{sec:markov_model}

Let $\Sigma_i = L(Y_1)+\cdots+L(Y_i)$ denote the total number of
source symbols consumed through the $i$-th phrase (with $\Sigma_0\dfn 0$,
consistent with the $\Sigma_n$ of Section~\ref{sec:background}), so
phrase $Y_i$ consists of the symbols
$X_{\Sigma_{i-1}+1},\dots,X_{\Sigma_i}$. We define the state of
the source at the $i$-th phrase boundary as the last symbol consumed
before that boundary, that is,
\begin{equation}
\label{eq:state_def}
S_i\dfn X_{\Sigma_i}, \qquad i=0,1,2,\dots
\end{equation}
and so $S_0=X_0$. Thus the state set is $\mathcal{X}$ itself, of
size $\alpha$; $S_{i-1}$ is the symbol in force when phrase $Y_i$
begins, and $S_i$ is the symbol reached when phrase $Y_i$ ends and the
parser resets to the root of the parsing tree.

The initial state $S_0$ is taken to have the stationary
distribution $\pi$ of the phrase-boundary chain $\{S_i\}_{i\ge0}$ just
defined (shown to be a genuine Markov chain in
Lemma~\ref{lem:markov_boundary} below) -- a distribution generally
different from the ordinary stationary distribution of
$\{X_k\}_{k\ge0}$ itself, since phrase boundaries sample states at
variable, state-dependent intervals rather than at every symbol
(Section~\ref{sec:markov_worked} gives a numerical instance of this
gap).

Since phrase $Y_i$ is parsed by running the parsing tree against the
symbols $X_{\Sigma_{i-1}+1},X_{\Sigma_{i-1}+2},\dots$ until a leaf is
reached, the triple $(L(Y_i),\ell(Y_i),S_i)$ is a measurable function
of $S_{i-1}=X_{\Sigma_{i-1}}$ together with the (as yet unconsumed)
continuation of the chain from time $\Sigma_{i-1}$ onward. This lets
the Markov property of $\{X_k\}_{k\ge0}$ pass through the parsing
recursion to the phrase level:

\begin{lemma}[Markov property at phrase boundaries]
\label{lem:markov_boundary}
Conditioned on $S_{i-1}$, the triple $(L(Y_i),\ell(Y_i),S_i)$ is
independent of $(S_0,\dots,S_{i-2},Y_1,\dots,Y_{i-1})$, with conditional
law depending on $S_{i-1}$ alone -- and, by time-homogeneity, the same
function of $S_{i-1}$ for every $i$, not merely free of the extra
history at each fixed $i$. Consequently $\{S_i\}_{i\ge0}$ is a
time-homogeneous Markov chain on $\mathcal X$.
\end{lemma}
\begin{proof}
$(L(Y_i),\ell(Y_i),S_i)$ is a function of $S_{i-1}=X_{\Sigma_{i-1}}$
and the continuation $(X_{\Sigma_{i-1}+1},X_{\Sigma_{i-1}+2},\dots)$
alone, where $\Sigma_{i-1}$ is a stopping time of $\{X_k\}_{k\ge0}$. By
the strong Markov property, conditioned on $X_{\Sigma_{i-1}}=S_{i-1}$
this continuation is independent of the past
$(X_0,\dots,X_{\Sigma_{i-1}})$ -- and so of
$$(S_0,\dots,S_{i-2},Y_1,\dots,Y_{i-1}),$$ 
functions of that past -- with
a law depending only on $S_{i-1}$, by time-homogeneity of the
transition kernel.
\end{proof}

This is what makes $\{S_i\}_{i\ge0}$ -- not $\{X_k\}_{k\ge0}$ itself -- the
Markov chain that matters for the moment analysis: the matrix
quantities below are indexed by $\mathcal X$ and evolve one
step per phrase rather than per symbol. The construction extends
verbatim to a $k$-th order Markov chain or a hidden finite-state
source, taking $S_i$ to be the last $k$ symbols or the hidden state;
we do not pursue this here.

\subsection{Matrix-valued moment generating functions}
\label{sec:markov_transforms}

For states $s,s'\in\mathcal{X}$, define the matrix-valued moment
generating function
\begin{equation}
[\boldsymbol{\mu}_j(t)]_{ss'} := \E\{[\ell(Y_i)]^je^{-tL(Y_i)}\mathbf{1}[S_i=s']\mid S_{i-1}=s\}, 
\qquad j=0,1,2,\dots,~~t\ge 0, \label{eq:muj_matrix_def}
\end{equation}
for any phrase index $i\ge1$ (by
Lemma~\ref{lem:markov_boundary}, this conditional law does not depend
on $i$). These are
$\alpha\times\alpha$ nonnegative matrices
parameterized by $t$: $\boldsymbol{\mu}_j(t)$ is the matrix analogue of the
scalar quantity $\mu_j(t)$ of Section~\ref{sec:background}, and
the case $j=0$ corresponds to the state transition
matrix at phrase boundaries, $[\boldsymbol{\mu}_0(t)]_{ss'} =
\E\{e^{-tL(Y_i)}\mathbf{1}[S_i=s']\mid S_{i-1}=s\}$. At $t=0$,
$[\boldsymbol{\mu}_0(0)]_{ss'}=\Pr\{S_i=s'\mid S_{i-1}=s\}$,
so $\boldsymbol{\mu}_0(0)$ is row-stochastic.

We assume throughout that
$\boldsymbol{\mu}_0(0)$ -- the phrase-boundary chain's own transition
matrix, as distinct from the transition matrix of $\{X_k\}_{k\ge0}$
itself -- is irreducible \emph{and aperiodic} (equivalently,
\emph{primitive}); this is a separate condition from the
irreducibility of $\{X_k\}_{k\ge0}$ assumed in Section~\ref{sec:markov}
(a tree can route every path ending in a given state through a
first symbol unreachable in one step from another state, so the two
irreducibilities do not simply hand each other over), and it holds for
every example in this paper. Aperiodicity is needed, not just
irreducibility: an irreducible but periodic chain can have several
eigenvalues tied for maximum modulus (e.g.\ the period-$2$ chain
$\left(\begin{smallmatrix}0&1\\1&0\end{smallmatrix}\right)$ has
eigenvalues $1$ and $-1$, both of modulus $1$), which would break the
spectral-gap argument used in Section~\ref{sec:markov_bias} and
Appendix~B to isolate the dominant eigenvalue's
contribution; primitivity is exactly what rules this out (by the
Perron--Frobenius theorem for primitive matrices, the Perron
eigenvalue is then \emph{strictly} greater in modulus than every other
eigenvalue).

A simple sufficient condition, easily checked without appeal to
Perron--Frobenius theory directly: if $\{X_k\}_{k\ge0}$ has a
\emph{fully positive} one-step transition matrix (i.e.\
$\Pr\{X_{k+1}=x'\mid X_k=x\}>0$ for every $x,x'\in\mathcal X$), then
$\boldsymbol\mu_0(0)$ has strictly positive entries throughout, and is
therefore automatically both irreducible and aperiodic. Indeed, full
positivity means any specific finite symbol sequence has positive
probability of occurring next, from any current state; since the
parsing tree is finite, every leaf corresponds to some such sequence,
so every leaf -- and hence every ending state -- is reachable with
positive probability from every starting state in a single phrase.
This condition is considerably stronger than needed (it fails, for
instance, whenever some transition is structurally forbidden), but it
covers most sources encountered in practice and requires no
computation beyond inspecting the source's own transition matrix.

Under this assumption, let $\pi$ denote the unique
stationary distribution of $\boldsymbol{\mu}_0(0)$ -- a row vector,
consistent with its left-multiplying matrices throughout (as in
$\pi\boldsymbol{\mu}_0(0)$ just below) -- satisfying
$\pi\boldsymbol{\mu}_0(0)=\pi$ and $\sum_{s}\pi_s=1$. For a
state-dependent single-phrase quantity $G(Y_i)$ (such as the phrase
length $L(Y_i)$, the codeword length $\ell(Y_i)$, or their product),
we write $\E_\pi\{G(Y_i)\} \dfn \sum_{s\in\mathcal{X}}\pi_s\E\{G(Y_i)\mid S_{i-1}=s\}$
for its expectation under the stationary regime; for $G=L$, this plays
the role
of $\bar L$ of Section~\ref{sec:background} in the Markov case, and is
the a.s.\ limit $\Sigma_n/n\to\E_\pi\{L(Y_i)\}$ by the ergodic theorem for
irreducible finite-state Markov chains.

\subsection{The first moment formula for the Markov case}

\begin{theorem}
\label{thm:markov}
If the initial state $S_0$ has the stationary distribution $\pi$ of 
the phrase-boundary chain, then
\begin{equation}
\E\{R_n\} = \int_0^\infty
\pi\left[\sum_{i=1}^n \boldsymbol{\mu}_0(t)^{i-1}\boldsymbol{\mu}_1(t)
\boldsymbol{\mu}_0(t)^{n-i}\right]\mathbf{1}\dif t,
\label{eq:markov}
\end{equation}
where $\mathbf{1}$ is the all-ones column vector and $\pi$ is understood to be
a row vector.
\end{theorem}

\begin{proof}
As in the proof of Theorem~\ref{thm:mean},
$\E\{R_n\}=\int_0^\infty\E\{\Lambda_ne^{-t\Sigma_n}\}\dif t$.
Expanding $\Lambda_ne^{-t\Sigma_n}=\sum_{i=1}^n \ell(Y_i)e^{-t(L(Y_1)+\dots+L(Y_n))}$:
\begin{equation}
\label{eq:markov_expand}
\E\{\Lambda_ne^{-t\Sigma_n}\} = \sum_{i=1}^n\E\bigl\{\ell(Y_i)e^{-tL(Y_i)}\prod_{j\ne i}e^{-tL(Y_j)}\bigr\}.
\end{equation}
For a given $i$, define $h_j\dfn e^{-tL(Y_j)}$ for $j\ne i$ and
$h_i\dfn\ell(Y_i)e^{-tL(Y_i)}$, so the $i$-th term of
\eqref{eq:markov_expand} is $\E\{\prod_{j=1}^n h_j\}$.
For $j=1,\dots,n+1$, define
\begin{equation}
\label{eq:phi_def}
\phi_j(s)\dfn\E\Bigl\{\prod_{k=j}^n h_k \Big| S_{j-1}=s\Bigr\}, \qquad s\in\mathcal X
\end{equation}
where it should be understood that the product is empty when $j=n+1$, so $\phi_{n+1}\equiv1$. This is
well-defined as a function of $s$ alone -- a conditional expectation
given a random variable is, by definition, always some function of
that variable -- with no appeal to the Markov property yet. What the
Markov property buys us is that these functions also satisfy the
\emph{backward recursion}
\begin{equation}
\label{eq:phi_recursion}
\phi_j(s) = \E\{h_j\,\phi_{j+1}(S_j)\mid S_{j-1}=s\}, \qquad j=n,n-1,\dots,1,
\end{equation}
i.e., that $\phi_j$ can be built from $\phi_{j+1}$ one phrase at a
time, rather than recomputing a fresh conditional expectation over
$k=j,\dots,n$ from scratch at every step. We prove
\eqref{eq:phi_recursion} via the following stronger claim, by
induction on $j$ decreasing from $n+1$ to $1$:
\begin{equation}
\label{eq:markov_ih}
\E\Bigl\{\prod_{k=j}^n h_k \;\Big|\; S_0,\dots,S_{j-1},Y_1,\dots,Y_{j-1}\Bigr\} = \phi_j(S_{j-1}),
\end{equation}
i.e., conditioning on the \emph{entire} history through phrase $j-1$,
rather than on $S_{j-1}$ alone as in \eqref{eq:phi_def}, changes
nothing. The case $j=n+1$ is immediate: both sides equal $1$, the
product on the left being empty and $\phi_{n+1}\equiv1$ by definition.

For the inductive step from $j+1$ to $j$ (assuming
\eqref{eq:markov_ih} at $j+1$, prove it at $j$), condition on
$(S_0,\dots,S_j,Y_1,\dots,Y_j)$ and use \eqref{eq:markov_ih} at $j+1$:
\begin{equation}
\E\Bigl\{\prod_{k=j}^n h_k\Big|S_0,\dots,S_{j-1},Y_1,\dots,Y_{j-1}\Bigr\}
= \E\Bigl\{h_j\,\phi_{j+1}(S_j)\Big|S_0,\dots,S_{j-1},Y_1,\dots,Y_{j-1}\Bigr\}.
\end{equation}
By Lemma~\ref{lem:markov_boundary} applied to phrase $j$,
$h_j\,\phi_{j+1}(S_j)$ -- a function of $(L(Y_j),\ell(Y_j),S_j)$ -- is
independent of $(S_0,\dots,S_{j-2},Y_1,\dots,Y_{j-1})$ given
$S_{j-1}$, so the right-hand side reduces to
$\E\{h_j\,\phi_{j+1}(S_j)\mid S_{j-1}\}$, a function of $S_{j-1}$
alone -- giving \eqref{eq:markov_ih} at $j$, and, comparing this same
computation against definition \eqref{eq:phi_def} of $\phi_j$, exactly
the recursion \eqref{eq:phi_recursion}. Only the \emph{one-step}
property of Lemma~\ref{lem:markov_boundary} is used, once per step of
the induction; the recursion is what accumulates it, phrase by phrase,
into \eqref{eq:markov_ih}.

Taking $j=1$ in \eqref{eq:markov_ih} gives
$\E\{\prod_{k=1}^n h_k\mid S_0=s\}=\phi_1(s)$ -- automatic from
\eqref{eq:phi_def} itself at $j=1$, and a useful consistency check;
averaging over
$S_0\sim\pi$ gives $\E\{\prod_{j=1}^n h_j\}=\sum_s\pi_s\phi_1(s)$.
Unwinding the recursion \eqref{eq:phi_recursion} and recognizing $[\boldsymbol\mu_0(t)]_{ss'}$,
$[\boldsymbol\mu_1(t)]_{ss'}$ from \eqref{eq:muj_matrix_def} at each step (the
matrix at step $j$ is $\boldsymbol\mu_1(t)$ if $j=i$ and $\boldsymbol\mu_0(t)$
otherwise),
\begin{equation}
\vec\phi_1 = [\boldsymbol\mu_0(t)]^{i-1}\boldsymbol\mu_1(t)[\boldsymbol\mu_0(t)]^{n-i}\,\mathbf1,
\end{equation}
where $\vec\phi_1$ is the column vector of values $\phi_1(s)$,
$s\in\mathcal X$. Hence the $i$-th term of \eqref{eq:markov_expand}
equals $\pi[\boldsymbol\mu_0(t)]^{i-1}\boldsymbol\mu_1(t)[\boldsymbol\mu_0(t)]^{n-i}\mathbf1$;
summing over $i=1,\dots,n$ and integrating over $t$ gives \eqref{eq:markov}.
\end{proof}

\subsection{Worked example and validation}
\label{sec:markov_worked}

We use the binary Markov source of Savari and Gallager
\cite{savari-gallager}: the state is the last bit emitted, and the
self-transition probability (probability of emitting the same bit as the
last one) is $q$. This source has two states $\{0,1\}$, with transition
probabilities $\Pr[S_i=s\mid S_{i-1}=s]=q$ and
$\Pr[S_i=1-s\mid S_{i-1}=s]=1-q$.

We use the parsing dictionary $\{00,01,1\}$ (the same as in
Section~\ref{sec:comparison}) and assign fixed-length codewords $\ell\equiv2$
(a V2F code). The three possible phrases and their contributions to
$\boldsymbol{\mu}_0(t)$ are:

\begin{itemize}
\item Phrase $1$ (one symbol equal to $1$): emitted from state $0$ with
probability $1-q$ (cross-transition) and from state $1$ with probability
$q$ (self-transition), driving the source to state $1$, with $L=1$.
\item Phrase $00$ (two $0$s): reached via symbol $0$ then another $0$;
probability from state $0$ is $q\cdot q=q^2$ and from state $1$ is
$(1-q)\cdot q=q(1-q)$, driving the source to state $0$, with $L=2$.
\item Phrase $01$: from state $0$ via $0$ (prob.\ $q$) then $1$
(prob.\ $1-q$), probability $q(1-q)$; from state $1$ via $0$ (prob.\
$1-q$) then $1$ (prob.\ $1-q$), probability $(1-q)^2$. Drives source to
state $1$, $L=2$.
\end{itemize}

Collecting these contributions into the matrix $\boldsymbol{\mu}_0(t)$
(rows indexed by starting state, columns by ending state):
\begin{equation}
\label{eq:mu0_matrix}
\boldsymbol{\mu}_0(t) = \begin{pmatrix}
q^2 e^{-2t} & (1-q)e^{-t}+q(1-q)e^{-2t}\\
q(1-q)e^{-2t} & qe^{-t}+(1-q)^2e^{-2t}
\end{pmatrix}.
\end{equation}
Since $\ell\equiv2$, $\boldsymbol{\mu}_1(t)=2\boldsymbol{\mu}_0(t)$ and \eqref{eq:markov} simplifies to
$\E\{R_n\}=2n\int_0^\infty\pi\boldsymbol{\mu}_0(t)^n\mathbf{1}\,\dif t$.

\begin{example}
\label{ex:markov}
Take $q=0.99$ (a highly persistent source). The stationary distribution
of parsing-point states is found by solving $\pi\boldsymbol{\mu}_0(0)=\pi$,
$\pi_0+\pi_1=1$, giving $\pi=(0.332,0.668)$ -- visibly different from
the ordinary stationary distribution of the symbol-level chain
$(X_k)_{k\ge0}$ itself, which is $(1/2,1/2)$ for \emph{every} $q$ by
the symmetry of its transition probabilities under $0\leftrightarrow1$;
this is the numerical instance, promised in
Section~\ref{sec:markov_model}, of the general fact that
phrase-boundary sampling need not preserve a chain's own stationary
law. Evaluating \eqref{eq:markov} by
numerical quadrature at $n=10$ phrases gives $\E\{R_{10}\}=1.6503$.
A Monte Carlo simulation of $2{,}000{,}000$ independent realizations of
the Markov parsing process gives
$\widehat{\E\{R_{10}\}}=1.6506\pm0.0003$, in agreement within one
standard error. Evaluating \eqref{eq:markov} at increasing $n$ (using the
$t=u/n$ substitution and full eigendecomposition of $\boldsymbol{\mu}_0(t)$ for
numerical stability at large $n$):
\begin{center}
\begin{tabular}{lcccc}
\hline
$n$ & $5$ & $10$ & $30$ & ($\infty$)\\
\hline
$\E\{R_n\}$ & $1.657$ & $1.650$ & $1.627$ & $\rho=1.497$\\
\hline
\end{tabular}
\captionof{table}{Finite-$n$ convergence of $\E\{R_n\}$ to the ergodic
limit $\rho$, for the Markov worked example of
Example~\ref{ex:markov}.}
\label{tab:markov_convergence}
\end{center}
The sequence decreases monotonically to the ergodic limit
$\rho=\E_\pi\{\ell\}/\E_\pi\{L\}=2/(\pi_0\cdot1.99+\pi_1\cdot1.01)=2/1.336=1.497$.
As a cross-check of the formula, write $\Lambda(t)$ for the
\emph{dominant} (Perron) eigenvalue of $\boldsymbol\mu_0(t)$ --
guaranteed simple and positive by the irreducibility of
$\boldsymbol\mu_0(0)$, and strictly greater in modulus than every
other eigenvalue by the additional aperiodicity assumed in
Section~\ref{sec:markov_transforms} -- with
$v(t)$ its corresponding right eigenvector and $u(t)$ its
corresponding left eigenvector, normalized so $u(t)v(t)\equiv1$. At
$t=0$, since $\boldsymbol\mu_0(0)$ is row-stochastic, $\Lambda(0)=1$
with $v(0)=\mathbf1$ (as $\boldsymbol\mu_0(0)\mathbf1=\mathbf1$) and
$u(0)=\pi$ (as $\pi\boldsymbol\mu_0(0)=\pi$, and indeed
$\pi\mathbf1=1$, consistent with the normalization). Differentiating
the eigenvalue equation $\boldsymbol\mu_0(t)v(t)=\Lambda(t)v(t)$ at
$t=0$ and left-multiplying by $u(0)=\pi$:
\begin{equation}
\pi\boldsymbol\mu_0'(0)v(0) + \pi\boldsymbol\mu_0(0)v'(0)
= \Lambda'(0)\,\pi v(0) + \Lambda(0)\,\pi v'(0);
\end{equation}
since $\pi\boldsymbol\mu_0(0)=\pi=\Lambda(0)\pi$, the second term on
each side is the same, $\pi v'(0)$, and cancels, leaving
$\Lambda'(0) = \pi\boldsymbol\mu_0'(0)\mathbf1$ (using $v(0)=\mathbf1$).
Since $[\boldsymbol\mu_0(t)\mathbf1]_s=\sum_{s'}[\boldsymbol\mu_0(t)]_{ss'}
=\E\{e^{-tL(Y_i)}\mid S_{i-1}=s\}$ (the indicator $\mathbf1[S_i=s']$
summed over $s'$ is just $1$), differentiating at $t=0$ gives
$\boldsymbol\mu_0'(0)\mathbf1=-g$ with $g_s:=\E\{L(Y_i)\mid S_{i-1}=s\}$,
so $\Lambda'(0)=-\pi g=-\sum_s\pi_s g_s=-\E_\pi\{L\}$ by the definition
of $\E_\pi\{L\}$ in Section~\ref{sec:markov_transforms} (this same $g$ reappears
as the reward vector of the Poisson equation in
Section~\ref{sec:markov_bias} below). Numerically,
$\Lambda'(0)=-2(q+1)/(2q+1)=-1.336$ (obtained directly by
differentiating the characteristic polynomial of $\boldsymbol\mu_0(t)$ at
$t=0$); indeed, $-\Lambda'(0)=1.336=\E_\pi\{L\}$, confirming the general
formula just derived.
\end{example}

\subsection{Asymptotic bias for Markov sources}
\label{sec:markov_bias}

The large-$n$ asymptotic analysis of \eqref{eq:markov} proceeds as in
Section~\ref{sec:laplace}: the substitution $t=u/n$ stabilizes the
numerical evaluation, and the resulting bias coefficient $n(\E\{R_n\}-\rho)$
converges to a finite limit as $n\to\infty$. This is 
governed by the dominant eigenvalue $\Lambda(t)$ of
$\boldsymbol\mu_0(t)$ and the reason is that, for large $n$,
$[\boldsymbol\mu_0(t)]^n$ itself is governed by $[\Lambda(t)]^n$ (times
a bounded projector term, standard for a matrix with a strictly
dominant eigenvalue) -- the matrix analogue of the scalar
$[\mu_0(t)]^n$ of the memoryless case, which trivially equals itself
raised to the $n$-th power. So extracting the $O(1/n)$ term from
\eqref{eq:markov} requires a Taylor series expansion of $\Lambda(t)$ around $t=0$
to the same order that Section~\ref{sec:laplace} needed for the
scalar quantity $f(t)=-\ln\mu_0(t)$. For
Example~\ref{ex:markov}, the results are displayed in Table
\ref{tab:markov_bias_convergence}.
\begin{center}
\begin{tabular}{lcccc}
\hline
$n$ & $2{,}000$ & $10{,}000$ & $50{,}000$ & ($n\to\infty$)\\
\hline
$\E\{R_n\}$ & $1.504$ & $1.499$ & $1.498$ & $1.497$\\
$n(\E\{R_n\}-\rho)$ & $12.2$ & $12.3$ & $12.4$ & $\approx12.4$\\
\hline
\end{tabular}
\captionof{table}{Numerical convergence of the bias coefficient
$n(\E\{R_n\}-\rho)$ to its limit, for
Example~\ref{ex:markov}.}
\label{tab:markov_bias_convergence}
\end{center}

A closed-form expression for the bias coefficient requires going one
order further than the eigenvalue derivative $\Lambda'(0)=-\E_\pi\{L\}$
already computed (Example~\ref{ex:markov}) -- exactly as the
memoryless case of Section~\ref{sec:laplace} needed not just
$f'(0)$ but also $f''(0)$. For a \emph{matrix} eigenvalue
problem, unlike the scalar case, getting the second-order term
$\Lambda''(0)$ requires first finding the first-order correction to
$\Lambda(t)$'s corresponding \emph{eigenvector} -- a standard fact of
matrix perturbation theory (the same mechanism as second-order
energy shifts needing first-order wavefunction corrections in
perturbation theory more generally) -- and this correction is what the
Poisson equation below solves for.

Specifically, consider the expansion $\boldsymbol\mu_0(t) = P + t\boldsymbol\mu_0'(0) + O(t^2)$
(with $P\dfn\boldsymbol\mu_0(0)$), $\Lambda(t)$'s corresponding right
eigenvector as
$v(t) = \mathbf{1} + t\,\be + O(t^2)$ (since $P\mathbf1=\mathbf1$, this
right eigenvector at $t=0$ is $\mathbf1$, and $\be$ is its unknown
first-order correction), and the eigenvalue itself as
$\Lambda(t) = 1 - \bar L t + O(t^2)$, where $\bar L:=\E_\pi\{L\}$.
Substituting into the eigenvalue equation
$\boldsymbol\mu_0(t)v(t)=\Lambda(t)v(t)$ and matching the coefficients
of $t$ on both sides gives
\begin{equation}
\label{eq:order_t}
(P-I)\be = -\bar L\mathbf1 - \boldsymbol\mu_0'(0)\mathbf1.
\end{equation}
Recall from Example~\ref{ex:markov} that
$[\boldsymbol\mu_0(t)\mathbf1]_s=\E\{e^{-tL(Y_i)}\mid S_{i-1}=s\}$, so
$\boldsymbol\mu_0'(0)\mathbf1=-g$ with $g_s\dfn\E\{L(Y_i)\mid S_{i-1}=s\}$
the expected phrase length given starting state $s$. Substituting into
\eqref{eq:order_t},
\begin{equation}
(P-I)\be = g-\bar L\mathbf1, \qquad\text{i.e.}\qquad (I-P)\be = \bar L\mathbf1-g.
\end{equation}
Since $(I-P)\mathbf1=0$, this equation only determines $\be$ up to an
additive multiple of $\mathbf1$; writing $w\dfn-\be$ and fixing this
freedom by the normalization $\pi w=0$ turns it into exactly the
\emph{Poisson equation}
\begin{equation}
\label{eq:Poisson}
(I-P)w = g - \E_\pi\{L\}\cdot\mathbf{1}, \qquad \pi w = 0,
\end{equation}
(this name and normalization are standard in Markov reward theory
\cite{meyn-tweedie},
where $(I-P)w=h$ with $\pi h=0$ is the equation for the
\emph{relative reward} $w$ associated with a per-state reward whose
$\pi$-average has already been subtracted off). Here $w$ -- the
relative-value vector -- is precisely
the relative-reward vector of Savari and Gallager \cite{savari-gallager},
computed there for the reward ``self-information of a phrase'' under
dictionary-size asymptotics; here the reward is ``phrase length $L$''
and the asymptotics are in phrase count.

This Poisson-equation correction is in fact the whole of what is
needed: a second, completely analogous Poisson equation for the
reward $\ell$, combined with a joint eigenvalue perturbation in both
the reward variable and $t$, yields the bias coefficient $C$ in closed
form for a general Markov V2V code -- not merely the V2F sub-case
tabulated above (writing $\Lambda_{\theta t}$ for the resulting mixed
second derivative, and $\Lambda''(0)$ for the second derivative of the
already-familiar $\Lambda(t)$ of Example~\ref{ex:markov}):

\begin{proposition}[Markov bias coefficient]
\label{prop:markov_bias}
For a Markov V2V code whose phrase-boundary chain is irreducible and
aperiodic (Section~\ref{sec:markov_transforms}), with stationary distribution
$\pi$, phrase-length mean $\bar L\dfn\E_\pi\{L\}$, codeword-length mean
$\bar\ell\dfn\E_\pi\{\ell\}$ (so $\rho=\bar\ell/\bar L$), and
$\bar c\dfn\E_\pi\{L\ell\}$,
\begin{equation}
\label{eq:markov_bias_main}
\lim_{n\to\infty}
n\bigl(\E\{R_n\}-\rho\bigr)= C=
\frac{\bar\ell\,\Lambda''(0) + \bar L\,\Lambda_{\theta t}}{\bar L^3},
\end{equation}
where
\begin{equation}
\Lambda''(0) = \pi\,\boldsymbol\mu_0''(0)\,\mathbf{1} - 2\,\pi\,\boldsymbol\mu_0'(0)\,w,
\qquad
\Lambda_{\theta t} = -\bar c - \pi\,\boldsymbol\mu_1(0)\,w + \pi\,\boldsymbol\mu_0'(0)\,w_\ell,
\end{equation}
$w$ solves \eqref{eq:Poisson} and $w_\ell$ solves the analogous equation
$(I-P)w_\ell = g_\ell - \bar\ell\,\mathbf{1}$, $\pi w_\ell=0$, with
$g_\ell(s)\dfn\E\{\ell(Y_i)\mid S_{i-1}=s\}$.
\end{proposition}

This is an extended version of the paper, including the complete
derivation of Proposition~\ref{prop:markov_bias} in Appendix~B below,
in place of the proof sketch given in the version submitted for
publication.

\begin{proof}
The argument extends the single-variable eigenvalue perturbation
above ($\Lambda'(0)=-\bar L$, via the Poisson equation for $w$) to two
variables: a second, completely analogous Poisson equation for the
reward $\ell$ produces $w_\ell$, and a joint perturbation of the
dominant eigenvalue in both the reward variable and $t$ together
produces $\Lambda''(0)$ and the new mixed term $\Lambda_{\theta t}$.
Appendix~B gives the full derivation of this
argument; the formula is verified numerically to 10+ significant
figures against direct high-precision computation of the exact
formula \eqref{eq:markov} at large $n$, on two different Markov V2V
codes, and collapses exactly to Proposition~\ref{prop:bias} in the
memoryless limit.
\end{proof}

\begin{example}[Closed form for Example~\ref{ex:markov}]
\label{ex:markov_bias_closed}
Applying Proposition~\ref{prop:markov_bias} to the V2F code of
Example~\ref{ex:markov} ($\ell\equiv2$, so $\boldsymbol\mu_1(t)=2\boldsymbol\mu_0(t)$
and $\bar c=2\bar L$) gives, after simplification,
\begin{equation}
\label{eq:C_v2f_closed}
C(q) = \frac{q(2q^2-q+1)}{2(1-q)(1+q)^3}.
\end{equation}
At $q=0.99$, $C(0.99)=12.3753$, matching the numerical target of the
table above (validated independently to 10 significant figures against
direct high-precision computation of \eqref{eq:markov} at $n$ up to
$10^9$).
\end{example}

\begin{example}[A genuine V2V code under the same Markov source]
\label{ex:markov_bias_v2v}
Proposition~\ref{prop:markov_bias} applies equally when $\ell$ is
itself random and state-dependent. Take the same Savari--Gallager
source and dictionary $\{00,01,1\}$, now equipped with the Huffman
codeword assignment of Section~\ref{sec:comparison},
$(\ell(00),\ell(01),\ell(1))=(1,2,2)$ -- a genuine V2V code, since both
$L$ and $\ell$ vary across phrases. Here $\ell$ depends only on the
destination state ($\ell=1$ when the parser returns to the root via
``$00$'', $\ell=2$ otherwise), which gives $\rho(q)=(3q+2)/[2(q+1)]$
and, applying Proposition~\ref{prop:markov_bias},
\begin{equation}
\label{eq:C_v2v_closed}
C(q) = \frac{q(4q^2+q+1)}{4(1-q)(1+q)^3}.
\end{equation}
At $q=0.99$: $\rho=1.2487$ and $C(0.99)=18.5623$, again validated to
10 significant figures against direct high-precision computation of
\eqref{eq:markov}, and against an independent check using a second
codeword assignment for which $\ell$ does \emph{not} depend only on the
destination state.
\end{example}

Figure~\ref{fig:markov_bias} plots both closed forms \eqref{eq:C_v2f_closed}
and \eqref{eq:C_v2v_closed} over the full range $q\in(0,1)$, rather than
at the single value $q=0.99$ tabulated above. Both bias coefficients
diverge as $q\to1$ -- a highly persistent source mixes slowly, and the
Poisson-equation solutions $w,w_\ell$ that drive $\Lambda''(0)$ and
$\Lambda_{\theta t}$ grow correspondingly large -- and, notably, the two
curves cross at exactly $q=1/3$ (as follows directly from equating
\eqref{eq:C_v2f_closed} and \eqref{eq:C_v2v_closed}): for weakly persistent sources the V2F
code has the larger finite-$n$ bias, while for strongly persistent
sources (including the $q=0.99$ case tabulated in both examples) the
genuine V2V code's bias is larger, mirroring the memoryless-case
tension already noted in Section~\ref{sec:comparison} between
optimising a code for $\rho$ and optimising it for $C$.

\begin{figure}[h]
\centering
\includegraphics[width=0.75\textwidth]{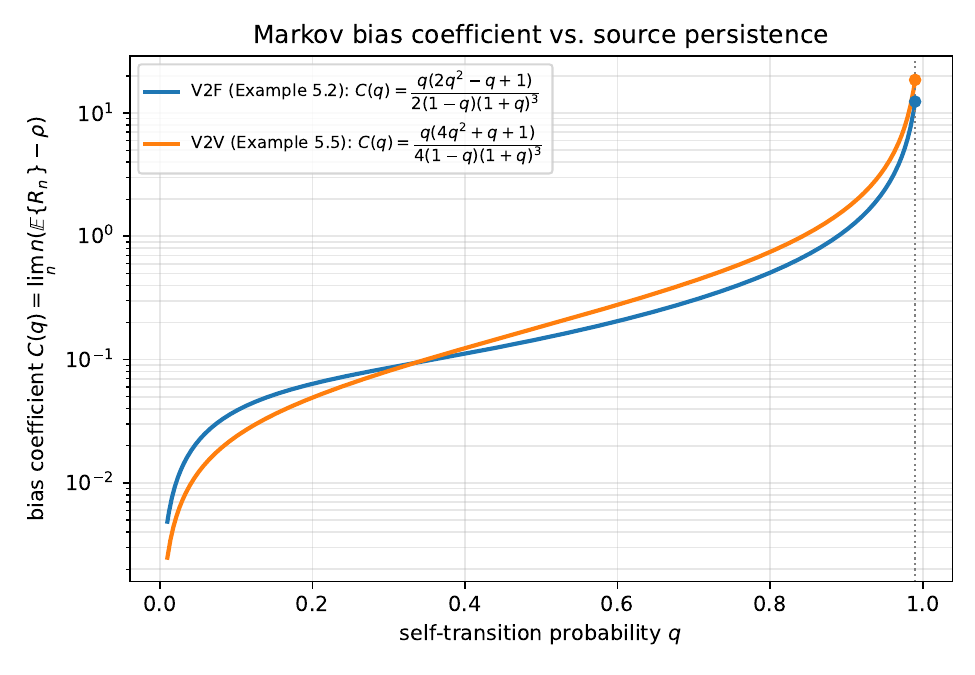}
\caption{The Markov bias coefficient $C(q)$ of
Proposition~\ref{prop:markov_bias}, for the V2F code of
Example~\ref{ex:markov_bias_closed} and the V2V code of
Example~\ref{ex:markov_bias_v2v}, as the self-transition probability
$q$ ranges over $(0,1)$. Markers indicate the $q=0.99$ values reported
in the text.}
\label{fig:markov_bias}
\end{figure}

\section{Conclusion}

The exact formula \eqref{eq:moments} gives, for any fixed V2V code applied to a
DMS, exact formulas for all integer moments of the realized
compression ratio $R_n$ at every finite $n$: mean, variance, and
skewness are computable by one-dimensional numerical quadrature, from
which an Edgeworth approximation to the CDF is obtained. The Laplace
method applied to \eqref{eq:moments} recovers the classical ratio-estimator
bias and variance asymptotics algebraically, by a different route than
the delta method (Section~\ref{sec:laplace}), and provides a design guideline (maximize $\Cov\{L,\ell\}$ for fastest
convergence) backed by a precise, finite-$n$ formula.

The extension
\eqref{eq:markov} to Markov sources replaces scalar quantities by matrix-valued
ones and is validated both against direct simulation and against the
correct ergodic limit; the $O(1/n)$ bias coefficient for this
Markov case is likewise obtained in closed form
(Proposition~\ref{prop:markov_bias}, Appendix~B, via
a joint eigenvalue perturbation that reduces exactly to
Proposition~\ref{prop:bias} in the memoryless limit.

Independently of the moment analysis, the observation that a V2V code is
an instance of a finite-state encoder lets us directly apply the
generalized Kraft inequality of \cite{merhav-gki}: via
Corollary~\ref{cor:lower}, this gives an explicit,
$m$-independent lower bound on the compression ratio in terms of the
dictionary parameters $M$, $\alpha$, and $\ell_{\max}$, with an $O(1/m)$
correction term that improves on the Ziv--Lempel bound whose analogous
constant grows linearly in $m$. The structural analysis of
Section~\ref{sec:comparison} decomposes the bias coefficient $C$ into a
$\Var\{L\}$ term and a $\Cov\{L,\ell\}$ term, each of which vanishes
identically for one of F2V and V2F; applied to the Khodak code, the
same decomposition shows that its improved $O(\bar L^{-5/3})$
redundancy and its correspondingly smaller bias coefficient are two
faces of the same design principle rather than independent
achievements.

\section*{Appendix A: Proof of the boundary Laplace lemma}
\label{app:laplace}
\renewcommand{\theequation}{A.\arabic{equation}}
\setcounter{equation}{0}

This appendix proves Lemma~\ref{lem:boundary_laplace}, the boundary
case of Laplace's method used in Section~\ref{sec:laplace}, in its
general form (arbitrary $k,m\ge1$).

\begin{proof}[Proof of Lemma~\ref{lem:boundary_laplace}]
Substitute $t=u/n$; then $t^{k-1}\dif t=u^{k-1}\dif u/n^k$, and
\begin{equation}
\frac{n!}{(n-m)!}\int_0^\infty t^{k-1}h(t)e^{-(n-m)f(t)}\dif t
= \frac{n!}{(n-m)!\,n^k}\int_0^\infty u^{k-1}h(u/n)e^{-(n-m)f(u/n)}\dif u.
\end{equation}
Now,
\begin{equation}
\dfrac{n!}{(n-m)!}=n(n-1)\cdots(n-m+1)=n^m\Bigl[1-\dfrac{m(m-1)}{2n}+O\left(\frac{1}{n^2}\right)\Bigr].
\end{equation}
Now, consider the Taylor series expansion of $-f$ around $t=0$: $-f(t)=-at+\frac b2t^2+O(t^3)$
as $t\to0$, so with $u$ fixed and $n\to\infty$,
\begin{equation}
-(n-m)f(u/n) =
(n-m)\Bigl[-\frac{au}{n}+\frac{bu^2}{2n^2}\Bigr]+O\left(\frac{1}{n^2}\right)
= -au + \frac{u}{n}\Bigl(ma+\frac b2 u\Bigr)+O\left(\frac{1}{n^2}\right),
\end{equation}
whence
\begin{equation}
e^{-(n-m)f(u/n)} = e^{-au}\Bigl[1+\frac un\Bigl(ma+\frac
b2u\Bigr)+O\left(\frac{1}{n^2}\right)\Bigr].
\end{equation}
Expanding $h$ as a Taylor series, we have
$h(u/n)=h_0+\dfrac{h_1u}{n}+O(1/n^2)$,
and multiplying the two
expansions, we obtain
\begin{equation}
u^{k-1}h(u/n)e^{-(n-m)f(u/n)}=e^{-au}\Bigl[u^{k-1}h_0 + \frac1n\bigl(mah_0u^k+\tfrac b2h_0u^{k+1}+h_1u^k\bigr)\Bigr]+O(n^{-2}).
\end{equation}
Integrating term by term using the identity $\int_0^\infty u^{k-1+j}e^{-au}\dif u=(k-1+j)!/a^{k+j}$,
\begin{equation}
\int_0^\infty u^{k-1}h(u/n)e^{-(n-m)f(u/n)}\dif u
= \frac{(k-1)!h_0}{a^k}+\frac1n\left(\frac{mh_0k!}{a^k}+
\frac{bh_0(k+1)!}{2a^{k+2}}+\frac{h_1k!}{a^{k+1}}\right)+O\left(\frac{1}{n^2}\right).
\end{equation}
Finally, multiplying by
$\frac{n!}{(n-m)!n^k}=n^{m-k}\bigl(1-\frac{m(m-1)}{2n}+O(1/n^2)\bigr)$
and collecting the two leading powers of $n$, the $O(1/n)$ coefficient
becomes
\begin{equation}
\frac{mh_0k!}{a^k}+\frac{bh_0(k+1)!}{2a^{k+2}}+\frac{h_1k!}{a^{k+1}}
-\frac{m(m-1)}{2}\cdot\frac{(k-1)!h_0}{a^k}
= \frac{(k-1)!}{2a^{k+2}}\Bigl(a^2h_0m(2k-m+1)+2ah_1k+bh_0k(k+1)\Bigr),
\end{equation}
using $2mk-m(m-1)=m(2k-m+1)$ to combine the $h_0$-terms. This is
exactly \eqref{eq:boundary_laplace}.
\end{proof}

\section*{Appendix B: Closed-form Markov bias coefficient}
\label{app:markov_bias}
\renewcommand{\theequation}{B.\arabic{equation}}
\setcounter{equation}{0}

This appendix derives the closed-form expression for
$n(\E\{R_n\}-\rho)$ used in Section~\ref{sec:markov}, for a general
Markov V2V code whose phrase-boundary chain is irreducible and
aperiodic (both $L$ and $\ell$ state-dependent random variables, not
merely the V2F sub-case of Example~\ref{ex:markov}). The argument
rests on a single general fact about how the Perron eigenvalue of a
matrix responds to two parameters at once (Lemma~\ref{lem:two_param}
below); once that lemma is established, the derivation is a short
substitution, not a lengthy one, and nothing is omitted.

\subsection*{Setup}

Define the joint matrix-valued moment generating function
\begin{equation}
\label{eq:joint_transform}
[\Psi(\theta,t)]_{ss'} := \E\{e^{\theta\ell(Y_i)-tL(Y_i)}\mathbf1[S_i=s']\mid S_{i-1}=s\},
\qquad s,s'\in\mathcal X,
\end{equation}
so that $\Psi(0,t)=\boldsymbol\mu_0(t)$ and
$\partial_\theta\Psi(\theta,t)|_{\theta=0}=\boldsymbol\mu_1(t)$. Let
$\Lambda(\theta,t)$ denote $\Psi(\theta,t)$'s dominant (Perron)
eigenvalue near the origin, so $\Lambda(0,t)=\Lambda(t)$ recovers the
scalar eigenvalue of Section~\ref{sec:markov_bias}.

\subsection*{A general two-parameter perturbation lemma}
\label{app:perturbation}

The remaining derivation is an instance of a single, general fact
about how the Perron eigenvalue of a matrix responds to two
parameters simultaneously -- stated and proved once here, in the
abstract, rather than worked out from scratch in the specific
notation of this problem.

\begin{lemma}[Two-parameter Perron eigenvalue perturbation]
\label{lem:two_param}
Let $A(\theta,t)$ be a matrix-valued function, jointly analytic near
the origin, with $A(0,0)=P$ primitive stochastic with Perron
eigenvalue $1$ and left/right eigenvectors $\pi,\mathbf{1}$ (normalized
$\pi\mathbf{1}=1$). Let $\Lambda(\theta,t)$, $v(\theta,t)$ denote the
Perron eigenvalue and right eigenvector of $A(\theta,t)$ near the
origin, normalized so $\pi\, v(\theta,t)\equiv1$ (so $v(0,0)=\mathbf1$,
$\Lambda(0,0)=1$). Write $A_x:=\partial_xA(0,0)$,
$A_{xy}:=\partial_x\partial_yA(0,0)$ for $x,y\in\{\theta,t\}$. Then
\begin{equation}
\label{eq:first_order_general}
\Lambda_x(0,0) = \pi A_x\mathbf1,
\end{equation}
and, writing $v_x$ for the (unique, once normalized by $\pi v_x=0$)
solution of the Poisson equation
\begin{equation}
\label{eq:poisson_general}
(P-I)v_x = \Lambda_x(0,0)\mathbf1 - A_x\mathbf1, \qquad \pi v_x=0,
\end{equation}
the second-order coefficients are
\begin{equation}
\label{eq:second_order_general}
\Lambda_{xy}(0,0) = \pi A_{xy}\mathbf1 + \pi A_xv_y + \pi A_yv_x.
\end{equation}
\end{lemma}

\begin{proof}
Differentiate the eigenvalue equation $A(\theta,t)v(\theta,t)=\Lambda(\theta,t)v(\theta,t)$
once in direction $x$ and evaluate at the origin:
$A_x\mathbf1+Pv_x = \Lambda_x(0,0)\mathbf1+v_x$, i.e.\
$(P-I)v_x=\Lambda_x(0,0)\mathbf1-A_x\mathbf1$, which is
\eqref{eq:poisson_general}; left-multiplying instead by $\pi$ (using
$\pi P=\pi$, $\pi\mathbf1=1$, and $\pi v_x=0$, the last from
differentiating the normalization $\pi v(\theta,t)\equiv1$ and
choosing this admissible value of the additive-$\mathbf1$ freedom in
$v_x$) gives $\pi A_x\mathbf1=\Lambda_x(0,0)$, which is
\eqref{eq:first_order_general}.

Differentiating once more, in direction $y$, and evaluating at the origin:
\begin{equation}
A_{xy}\mathbf1+A_xv_y+A_yv_x+Pv_{xy} = \Lambda_{xy}(0,0)\mathbf1+\Lambda_x(0,0)v_y+\Lambda_y(0,0)v_x+v_{xy}.
\end{equation}
Left-multiplying by $\pi$: the $Pv_{xy}$ and $v_{xy}$ terms both
become $\pi v_{xy}$ (using $\pi P=\pi$) and cancel; the
$\Lambda_x(0,0)v_y$ and $\Lambda_y(0,0)v_x$ terms vanish (using
$\pi v_x=\pi v_y=0$); what remains is exactly
\eqref{eq:second_order_general}.
\end{proof}

\paragraph{Application.} Take $A=\Psi$. Differentiating
\eqref{eq:joint_transform} gives $A_t\mathbf1=-g$, $A_\theta\mathbf1=g_\ell$
(where $g_\ell(s):=\E\{\ell(Y_i)\mid S_{i-1}=s\}$), $A_{tt}=\boldsymbol\mu_0''(0)$,
$A_\theta=\boldsymbol\mu_1(0)$, $A_t=\boldsymbol\mu_0'(0)$, and
$A_{\theta t}\mathbf1=-g_{L\ell}$ where
$g_{L\ell}(s):=\E\{L(Y_i)\ell(Y_i)\mid S_{i-1}=s\}$ (differentiating
the exponent $\theta\ell(Y_i)-tL(Y_i)$ once in each variable brings
down a factor $-L(Y_i)\ell(Y_i)$). By
\eqref{eq:first_order_general}, $\Lambda_t(0,0)=-\pi g=-\bar L$ and
$\Lambda_\theta(0,0)=\pi g_\ell=\bar\ell$, recovering the already-known
values. The two Poisson equations \eqref{eq:poisson_general}, for
$x=t$ and $x=\theta$ respectively, are
\begin{equation}
\label{eq:poisson_L}
(I-P)w = g-\bar L\mathbf{1}, \qquad \pi w=0
\qquad(w:=-v_t),
\end{equation}
\begin{equation}
\label{eq:poisson_ell}
(I-P)w_\ell = g_\ell-\bar\ell\,\mathbf{1}, \qquad \pi w_\ell=0
\qquad(w_\ell:=v_\theta),
\end{equation}
matching the Poisson equation of Section~\ref{sec:markov_bias}, with
the sign flip for $w$ tracing to $A_t\mathbf1=-g$ carrying an extra
minus sign that $A_\theta\mathbf1=g_\ell$ does not. Substituting into
\eqref{eq:second_order_general} with $(x,y)=(t,t)$ and $(x,y)=(\theta,t)$:
\begin{equation}
\label{eq:Lambda_pp_general}
\Lambda''(0) = \pi\,\boldsymbol\mu_0''(0)\,\mathbf{1} - 2\,\pi\,\boldsymbol\mu_0'(0)\,w,
\qquad
\Lambda_{\theta t} = -\bar c - \pi\,\boldsymbol\mu_1(0)\,w + \pi\,\boldsymbol\mu_0'(0)\,w_\ell,
\end{equation}
where $\bar c:=\E_\pi\{L\ell\}=\pi g_{L\ell}$ -- exactly the formulas
quoted in Proposition~\ref{prop:markov_bias}, obtained here by pure
substitution into Lemma~\ref{lem:two_param} rather than a fresh
derivation. The resulting formula was additionally verified
numerically to 10+ significant figures against direct high-precision
computation of the exact formula \eqref{eq:markov} at large $n$, on
two different Markov V2V codes built on the source of
Example~\ref{ex:markov} (one with $\ell$ depending only on the
destination state, one without that special structure), and collapses
exactly to Proposition~\ref{prop:bias} in the memoryless limit.

\end{document}